\def\year{2021}\relax
\documentclass[letterpaper]{article} 
\usepackage{aaai21}  
\usepackage{times}  
\usepackage{helvet} 
\usepackage{courier}  
\usepackage[hyphens]{url}  
\usepackage{graphicx} 
\usepackage{natbib}  
\usepackage{caption} 
\usepackage{microtype}      
\usepackage{multirow}

\usepackage{times}  
\usepackage{helvet}  
\usepackage{courier}  
\usepackage{url}  
\usepackage{graphicx}  
\usepackage{soul}
\usepackage{url}

\usepackage{stackengine}
\usepackage{amsfonts}
\usepackage{amsmath}
\usepackage[switch]{lineno}
\usepackage[ruled,vlined, linesnumbered, boxed]{algorithm2e}
\usepackage{amsthm}
\usepackage{amssymb}
\usepackage{mathtools}
\usepackage{xcolor}
\newtheorem{Example}{Example}
\newtheorem{Theorem}{Theorem}
\newtheorem{Lemma}{Lemma}
\newtheorem{Corollary}{Corollary}
\newtheorem{Proposition}{Proposition}

\title{Computing \textit{\textbf{Ex Ante}} Coordinated Team-Maxmin Equilibria in Zero-Sum Multiplayer Extensive-Form Games}

\author{
    Youzhi Zhang, Bo An, Jakub {\v{C}}ern{\'y}
}

\affiliations{
   \\
   School of Computer Science and Engineering, Nanyang Technological University, Singapore \\
   \{yzhang137, boan\}@ntu.edu.sg, cerny@disroot.org
}

\begin{document}

\maketitle
\begin{abstract}
Computational game theory has many applications in the modern world in both adversarial situations and the optimization of social good. While there exist many algorithms for computing solutions in two-player interactions, finding optimal strategies in multiplayer interactions efficiently remains an open challenge. This paper focuses on computing the multiplayer Team-Maxmin Equilibrium with Coordination device (TMECor) in zero-sum extensive-form games. TMECor models  scenarios when a team of players coordinates {\it ex ante} against an adversary. Such situations can be found in card games (e.g., in Bridge and Poker), when a team works together to beat a target player but communication is prohibited; and also in real world, e.g., in forest-protection operations, when coordinated groups have limited contact during interdicting illegal loggers. The existing algorithms struggle to find a TMECor efficiently because of their high computational costs. To compute a TMECor in larger games, we make the following key contributions: (1) we propose a hybrid-form strategy representation for the team, which preserves the set of equilibria; (2) we introduce a column-generation algorithm with a guaranteed finite-time convergence in the infinite strategy space  based on a novel best-response oracle; (3) we develop an associated-representation technique for the exact representation of the multilinear terms in the best-response oracle; and (4) we experimentally show that our algorithm is several orders of magnitude faster than prior state-of-the-art algorithms in large games.
\end{abstract}

\section{Introduction}
\label{TMECorIntro}
One of the most important problems in artificial intelligence  is to design algorithms for agents who make complex decisions in interactive environments  \cite{russell2016artificial}.
So far, researchers made significant progress mostly in non-cooperative two-player games, focusing on finding a Nash Equilibrium (NE) \cite{nash1951non,von1996efficient,zinkevich2008regret} or a Stackelberg equilibrium \cite{conitzer2006computing}. These results paved the way for many applications, such as in security games that have high social impact \cite{sinha2018stackelberg} or poker algorithms that defeated top human professionals \cite{moravcik2017,brown2018superhuman}. However, the research in multiplayer games remains limited. Theoretical results were achieved only for games with specific structures (e.g., polymatrix games \cite{cai2011minmax}), or there are no theoretical guarantees at all (e.g., the algorithm in \citeauthor{brown2019superhuman} \citeyearpar{brown2019superhuman}). Finding and playing NEs in multiplayer games is difficult due to the following two reasons. First, computing NEs is PPAD-complete even for three-player zero-sum games \cite{chen20053}. And second, NEs are neither unique nor exchangeable in multiplayer games, 
which makes it almost impossible for the players who choose their strategies independently to form an NE together. The research on multiplayer games hence focuses on alternative solution concepts with more favorable properties.

Team-Maxmin Equilibrium with Coordination device  (TMECor, TMEsCor as a plural) \cite{celli2018computational,farina2018ex} is a solution concept that models a situation  when a  team  of players shares the same utility function and coordinates {\it ex ante} against an adversary. 
That is, the team members are allowed to discuss and agree on tactics before the game starts, but they cannot communicate during the game.
Celli and Gatti \citeyearpar{celli2018computational} show that {\it ex ante} coordination can be modelled using a coordination device, assuming that the adversary does not observe any signal from the device. The team members agree on a planned strategy (e.g., a mixed strategy) in the planning phase, and then, just before the game starts, the coordination device randomly picks a pure joint  strategy (from the planned strategy) for the team members to act upon.  
A TMECor is an NE between the team (i.e., each team member has no incentive to deviate)  and the adversary in a zero-sum multiplayer extensive-form game, and it has the properties of NEs in zero-sum two-player games (e.g., exchangeability). The study of TMECor is motivated by its ability to capture many real-world scenarios. 
For example, in multiplayer poker games, a team may play against an adversary player, but they cannot communicate and discuss their strategy during the game due to the rules. In Bridge, when the game reaches the phase of the play of the hand, two defenders who form a team play against the declarer. Or in security games in which several different groups (e.g., NGOs, the Madagascar National Parks, local police, and community volunteers) aim to protect forests from illegal logging \cite{mccarthy2016preventing},  TMECor models the groups' inability to communicate while they try to interdict the escaping loggers.

TMECor has better properties than NE in multiplayer games; however, computing it is still difficult---it was shown to be FNP-hard \cite{celli2018computational}. The problem can be formulated as a linear program \cite{celli2018computational}, where the team plays joint normal-form strategies for all team members, but each member's normal-form strategy space is exponential in the size of the game tree. A Column Generation (CG) algorithm was hence proposed to compute a solution more efficiently \cite{celli2018computational}. The most important component of the algorithm is a Best Response  Oracle (BRO) that computes an optimal strategy of the team against the adversary's strategy, but it is in itself an NP-hard problem \cite{celli2018computational}. A BRO can be formulated as a Mixed-Integer Linear Program (MILP) that involves a large number of integer variables \cite{celli2018computational,farina2018ex}. As a consequence, the existing approaches fail to scale up to larger scenarios (see Section \ref{sectionTMECorRelatedwork} for  details).

{\bf Main Contributions.} The most significant outcome of our work is a new algorithm for computing a TMECor, which runs several  orders of magnitude faster than the state-of-the-art algorithms and scales to much larger games. To design this algorithm, we make several key contributions. The first contribution is a new hybrid-form strategy representation for the team's strategy in a TMECor. Based on this representation, we develop a CG method that guarantees convergence to a TMECor in a finite number of iterations, despite the fact that the space of our hybrid-form strategies is infinite. The core component of the CG method is a novel BRO. Our BRO is formulated as a multilinear program in which the multilinear terms represent reaching probabilities for terminal nodes to reduce the number of involved integer variables. We show that the BRO can be transformed into an MILP exactly using another contribution: a novel global optimization technique called Associated Representation Technique (ART). Another essential property of ART is that it efficiently generates associated constraints for the equivalence relations between multilinear terms, which significantly speeds up the computation of the BRO's MILP formulation by reducing its space of feasible solutions. All together, our approach shows that formulating the problem as a multilinear program with global optimization techniques can be significantly faster than the direct formulation as a linear program.

\section{Preliminaries}\label{TMECorPre}
An imperfect-information extensive-form game (EFG) is defined by a tuple $(N,A,H,L,\chi,\rho,u,I)$ \cite{shoham2008multiagent}. $N=\{1,\dots,n\}$ denotes a finite set of  players and $A$ is a finite set of actions. $H$ is a finite set of nonterminal decision nodes (sequences of actions (histories)) in the game, with $L$ being a set of leaf (terminal)  nodes. To each nonterminal decision node the function $\chi: H\rightarrow 2^A$ assigns a subset of possible actions to play, while function $\rho: H\rightarrow N\cup\{c\}$ identifies an acting player ($c$ denotes chance). Moreover, we denote $H_i=\{h\mid \rho(h)=i,h\in H\}, \forall i\in N$. 
To determinate the outcomes, we use $u=(u_1,\dots,u_n)$, where $u_i:L\rightarrow \mathbb{R}$ is player $i$'s utility function assigning a finite utility to each terminal node. The imperfect observations are modelled through the set of information sets $I=(I_1,\dots,I_n)$. $I_i$ is the set of player $i$'s information sets (a partition of $H_i$), such that $\rho(h_1)=\rho(h_2)$ and $\chi(h_1)=\chi(h_2)$ for any $I_{i,j}\in I_i$ with $h_1, h_2\in I_{i,j}$. We assume that actions are unique to information sets, i.e., there exists only one information set $I_{i,j}$ such that $a\in \chi(I_{i,j})$ for any $a\in A$.

A sequence $\sigma_i \in \Sigma_i$ is an ordered list of actions taken by a single player $i$ leading to some node $h$. $\varnothing$~stands for the empty sequence (i.e., a sequence with no actions).  We use $\text{seq}_i(I_{i,j})$ to denote the player $i$'s sequence leading to $I_{i,j}\in  I_i$, and $\text{seq}_i(h)$ for the player $i$'s sequence leading to $h\in  H\cup L$. We assume perfect recall, i.e., for each player $i$ and nodes $h_1,h_2\in I_{i,j} \in I_i$,  $\text{seq}_i(h_1)=\text{seq}_i(h_2)$. A realization plan ({\bf sequence-form strategy}, also representing a behavioral strategy) of player $i$ is a function $r_i:\Sigma\rightarrow [0,1]$ satisfying the network-flow constraints:
\begin{subequations}
\begin{align}
  r_i(\varnothing)&=1\label{sqconstraint1ad}\\
   \sum_{a\in \chi(I_{i,j})}\!  r_i(\sigma_i a)&=r_i(\sigma_i) \quad \forall I_{i,j}\!\in\! I_i, \sigma_i\!=\!\text{seq}_i(I_{i,j})  \label{sqconstraintsecondad}\\
  r_i(\sigma_i)&\geq 0 \quad \forall  \sigma_i\in \Sigma_i.\label{sqconstraint2ad}
\end{align}
\end{subequations}
Let $\mathcal{R}_i$ be the set of all (mixed) sequence-form strategies.
We call $r_{i}$   a pure sequence-form strategy if $r_{i}(\sigma_i)\in\{0,1\}$  for all $\sigma_{i}\in \Sigma_{i}$. The set of pure sequence-form strategies is denoted as $\overline{\mathcal{R}}_{i}$ and naturally, we have $\overline{\mathcal{R}}_{i}\subseteq \mathcal{R}_i$. 

A pure {\bf normal-form strategy} of player $i$ is a tuple $\pi_i\in \Pi_i =\times_{I_{i,j}\in I_i}\chi(I_{i,j})$ specifying one action to play in each information set of player $i$. In EFGs, the size of $\Pi_i$ is exponential in the size of the game tree. A reduced normal-form strategy specifies actions only in reachable information sets due to earlier actions. Henceforth, we focus on reduced normal-form strategies (despite their size being still exponential in the size of the game tree) and refer to them as normal-form strategies. A mixed normal-form strategy $x_i$ is a probability distribution over $\Pi_i$, denoted as $x_i\in\Delta(\Pi_i)$. For any pure (mixed) normal-form strategy there exists an equivalent pure (mixed) sequence-form strategy \cite{von1996efficient}, as shown in Example   \ref{TMECorExampleNotation}. 
We call two strategies of player $i$ realization-equivalent (using notation $\sim$) if they induce the same probabilities for reaching nodes for all strategies of other players. Indeed, two strategies are realization-equivalent if and only if they correspond to the same realization plan \cite{von1996efficient}.  

\begin{figure*}
    \centering
    \includegraphics[height=4cm]{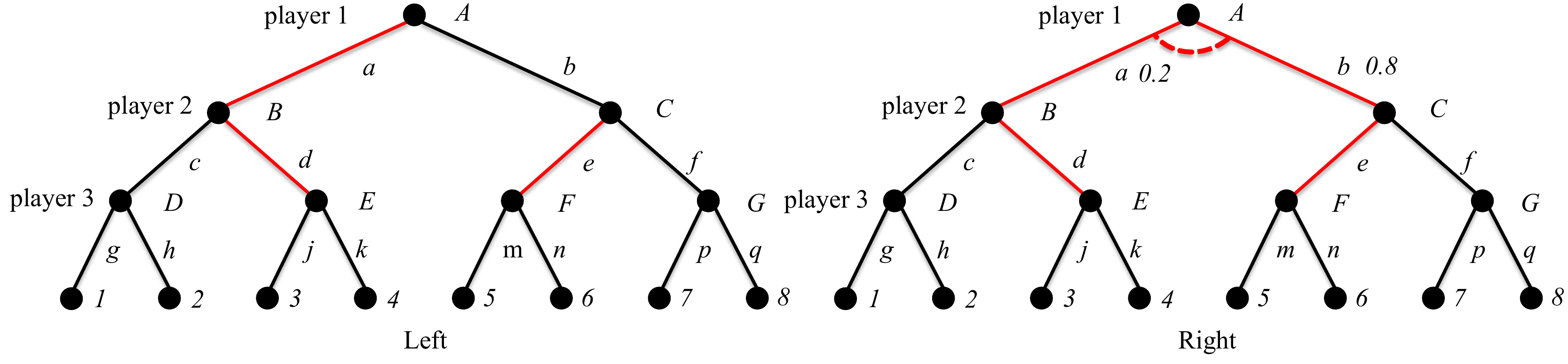}
    \caption{An example of an EFG with 3 players. Player 1 acts in information set $A$, player 2 acts in information sets $B$ and $C$, and player 3, assuming the role of the adversary, acts in information sets $D$--$G$. Actions in the information sets are denoted $a$--$q$, and they may also represent sequences in this case. Nodes $1$--$8$ are terminal nodes. Each side shows a different team's strategy.
 }
    \label{TMECorNotationExample}
\end{figure*}

{\bf Team-Maxmin Equilibrium with Coordination Device} \cite{celli2018computational,farina2018ex} is a solution concept that models a scenario when a single team $T=\{1,\dots,n-1\}$ with {\it ex ante} coordination plays against an adversary $n$. We assume the team shares the same utility $u_i(l)=u_j(l), \forall i,j\in T, l\in L$ and the utility of the adversary is $u_n(l) = -u_T(l)=-\sum_{i\in T}u_i(l),\forall l\in L$, i.e., it is a zero-sum EFG. The {\it ex ante} coordination means that the team players can communicate only before the game starts, through a coordination device. A pure strategy of the team in TMECor is represented by a joint normal-form strategy $\pi_T\in\Pi_T=\times_{i\in T}\Pi_i$. A mixed strategy $x_T$ is then a probability distribution over $\Pi_T$, i.e., $x_T\in\Delta(\Pi_T)$.  We use $L_{\pi_T,\sigma_n}\subseteq L$ to identify a set of terminal nodes reachable by strategy profile $(\pi_T,\sigma_n)$, as shown in Example \ref{TMECorExampleNotation}. Then, the extended utility function of the team's pure strategy $\pi_T$ specifies the utility of profile $(\pi_T,\sigma_n)$ due to chance nodes as $U_T(\pi_T,\sigma_n) = \sum_{l\in L_{\pi_T,\sigma_n}}u_T(l)c(l)$, where $c(l)$ denotes the chance probability of reaching $l$. For the team's mixed strategy $x_T$, we formulate the extended expected utility function as:
  \begin{align}\label{U_T_X_sigma}\textstyle U_T(x_T,\sigma_n)= \sum_{\pi_T\in \Pi_T}U_T(\pi_T,\sigma_n)x_T(\pi_T).\end{align} 
 Using a realization plan of the adversary, we write $L_{\pi_T,r_n}$ to denote a set of terminal nodes reachable by strategy profile $(\pi_T,r_n)$, as shown in Example   \ref{TMECorExampleNotation}, and  $U_T(\pi_T,r_n) = \sum_{l\in L_{\pi_T,r_n}}r_n(\text{seq}_n(l))u_T(l)c(l)$. For a mixed strategy,  
\begin{align}\label{U_T_X_r}\textstyle U_T(x_T,r_n)= \sum_{\pi_T\in \Pi_T}U_T(\pi_T,r_n)x_T(\pi_T).\end{align}
\begin{Example}    \label{TMECorExampleNotation}
On the left side of Figure \ref{TMECorNotationExample} we show an example of a strategy of a team consisting of players 1 and 2. $\pi_1=a$ is a normal-form strategy of player 1, in which player 1 takes action $a$ in his unique information set. $\pi_2=de$ is then a normal-form strategy of player 2, such that player 2 takes actions $d$ and $e$ in his two information sets. $\pi_1$ is equivalent to a pure sequence-form strategy $r_1$ with $r_1(a)=1$, and $\pi_2$ is equivalent to a pure sequence-form strategy $r_2$ with $r_2(d)=r_2(e)=1$. Together, strategies $\pi_1$ and $\pi_2$ form a pure joint normal-form strategy of the team $\pi_T=(\pi_1,\pi_2)$. Given the adversary's sequence $\sigma_3=j$, the set of terminal nodes reachable by $(\pi_T,\sigma_3)$ is $L_{\pi_T,\sigma_3}=\{3\}$. In case that $\sigma_3=m$, we have $L_{\pi_T,\sigma_3}=\emptyset$. Given any   sequence-form strategy $r_3$ of the adversary, we have $L_{\pi_T,r_3}=\{3,4\}$.
\end{Example}

A TMECor $(x_T,r_n)$ is a Nash Equilibrium (NE) (i.e., $x_T\in \arg\max_{x_T}U_T(x_T,r_n)$ and $r_n\in \arg\max_{r_n}-U_T(x_T,r_n)$) in which the team is treated as a single player. Therefore, we have $x_T\in \arg\max_{x_T} \min_{r_n}$ $U_T(x_T,r_n)$ due  to the zero-sum assumption. Because the team members share the same utility, none of them has an incentive to deviate. A strategy profile is an $\epsilon$-TMECor if neither the team nor the adversary is to gain more than $\epsilon$ if one of them deviates.

 \section{ Related Work}\label{sectionTMECorRelatedwork}
The Team-Maxmin Equilibrium (TME)~\cite{von1997team,celli2018computational,zhang2020converging} is a solution concept closely related to TMECor, in which a  team of players with the same utility function plays against an adversary independently, without coordination. Contrary to TMECor, the team members in TME are assumed to use behavioral strategies. Previous literature identified two main issues associated with TME. First, computing a TME, i.e., finding the optimal joint behavioral strategies of team members, is a non-convex FNP-hard optimization problem \cite{celli2018computational}. Second, the equilibrium strategies of TME might be significantly suboptimal compared to TMECor because of the lack of coordination~\cite{celli2018computational,farina2018ex}. Another complication arises if we attempt to coordinate the players using behavioral strategies. Due to the lack of communication between team members during the game, the team experiences imperfect recall. Because of imperfect recall, the behavioral strategies are not realization-equivalent to normal-form strategies induced by the coordination device, and thus can not capture the correlation between the team members' normal-form strategies. This was also shown to potentially result in considerable losses of utility to the team \cite{farina2018ex}. Moreover, with imperfect recall, there is even no guarantee for the existence of an NE in behavioral strategies \cite{wichardt2008existence}. Using normal-formal strategies is hence crucial to TMECor.

Because the number of normal-form strategies is exponential in the size of the game, Celli and Gatti  \citeyearpar{celli2018computational} proposed a hybrid representation of players' strategies to speed up the computation of TMEsCor. In their representation, only the adversary uses sequence-form strategies. Their CG algorithm for TMEsCor formulates a BRO using an MILP with $|L|$ integer variables. However, the algorithm does not scale well as $|L|$ can be extremely large for games of moderate size (see Table \ref{kuhnLeducpokerTMECorResults}). In an attempt to overcome this issue, Farina  et al. \citeyearpar{farina2018ex} developed a realization-form strategy representation that focuses on the probability distribution on terminal nodes of the game tree. Their representation is then used to derive an auxiliary game that represents the original game in the form of a set of subtrees. In each subtree, the adversary faces one team member, while they both use realization-form strategies. Rather than to derive a TMECor-specific strategy representation, their approach is to create a hybrid-form tree structure for the team.  An essential drawback of this approach is that the auxiliary game is based on the realization-form strategy that is not executable from the game tree's root. It requires a cumbersome reconstruction algorithm to apply it \cite{celli2019learning}. On the other hand, the auxiliary game enables to formulate a faster BRO with a number of integer variables equal to $\sum_{i\in T\setminus\{1\}}|\Sigma_i|$. Despite being a clear improvement over the approach of Celli and Gatti \citeyearpar{celli2018computational}, the BRO of Farina  et al. \citeyearpar{farina2018ex} is still inefficient in larger games with a higher number of sequences. Moreover, this BRO is not compatible with associated constraints \cite{zhang2020computing} that were shown to significantly improve the scalability by reducing the feasible solution space of an MILP. In this paper, we develop a hybrid-form strategy representation inspired by the previous literature, but which  does not depend on the inexecutable realization-form strategies. Simultaneously, the representation gives rise to a novel BRO that considerably reduces the number of integer variables and is compatible with associated constraints. Moreover, we design a new global optimization technique called ART that does not depend on the recursive generation of associated constraints proposed by Zhang and An \citeyearpar{zhang2020computing}. ART also works for games with four or more players  and    significantly improves the scalability in experiments. Note that  the concurrent work \cite{farina2020faster}   solves games with only three players.   We give more details on the related work in Appendix~\ref{RelatedWrokDetail}.

 \begin{table*}
 \footnotesize
  \centering
   \begin{tabular}{c|l|c|l }
     \hline
     
$\mathcal{R}_i$&set of player $i$'s mixed sequence-form strategies including $r_i$&$\mathcal{F}_T$& set of pure hybrid-form strategies including $f_T$\\
$\overline{\mathcal{R}}_i$&set of   player $i$'s pure sequence-form strategies  &$\Delta(\Pi_T)$&set of mixed strategies  including $x_T$\\
        $\Pi_{T}$&set of the team's   joint pure normal-form strategies including $\pi_{T}$  & $\Delta(\mathcal{F}_T)$&set of mixed strategies  including $\overline{x}_T$ \\
         $T\setminus{1}$&set of the team members except player 1   &$\sim$&realization equivalence  \\

              $\sigma_T(i)$&the team player $i$'s sequence in the joint sequence $\sigma_T$  & $\Sigma_i$&set of player $i$'s   sequences including $\sigma_i$ \\

              seq$_i(*)$  &player $i$'s sequence reaching  a node/information set $*$ & $\Sigma_T$&set of the team's joint sequences including $\sigma_T$  \\
               $w(\sigma_T)$&   variable representing multilinear term $\prod_{i\in T}r_i(\sigma_T(i))$& $S_x$& support set of  mixed strategy $x$ with the size $|S_x|$ \\
      \hline
   \end{tabular}
       \caption{The notation used in Section \ref{TMECorApproach}, in addition to the standard EFG notation.
       }
\end{table*}
       
\section{Hybrid-Form Column Generation with Efficiently Solvable Multilinear Oracle}\label{TMECorApproach}

Now we describe our novel approach for computing TMECor efficiently. First, we introduce a new strategy representation for the team, and we show how to use this representation to compute an exact TMECor. Then we propose a column-generation algorithm based on our representation, which iteratively calls a multilinear BRO. Finally, we develop a novel global optimization technique to solve our BRO efficiently.

\subsection{Hybrid-Form Strategies for the Team}
 In our hybrid-form team strategy representation, one team member acts according to a (mixed) sequence-form strategy, while the other members use pure normal-form strategies.\footnote{Note that all team players are free to use pure normal-form strategies--in that case, our algorithm still works, but it would significantly slow down the computation when compared with the hybrid-form strategies, as shown in the experiments. In contrast, if two or more team players use sequence-form strategies, the reaching probabilities for terminal nodes are no longer exactly representable by linear constraints (see Eqs.(\ref{werconstraint1}) and  (\ref{werconstraint2})).} Any team player can play the sequence-form strategy, and without  loss of generality, we opt for player 1. A pure hybrid-form strategy of the team is defined by a tuple\footnote{To differentiate between $f_T$ and $\overline{x}_T$ (a probability distribution over the space of $f_T$  defined later), we call $f_T$ a ``pure'' hybrid-form strategy and $\overline{x}_T$ a ``mixed'' hybrid-form strategy.}:  
\begin{align*}f_T=(r_{1},\pi_{T\setminus\{1\}}),\end{align*} 
where  $r_{1}\in \mathcal{R}_{1}$, and $\pi_{T\setminus\{1\}}\in \Pi_{T\setminus\{1\}}=\times_{i\in T\setminus\{1\}}\Pi_i$, as shown in Example \ref{TMECorExampleNotation2}.  We denote $\mathcal{F}_T$ the set of pure hybrid-form strategies and refer to $\overline{x}_T\in\Delta(\mathcal{F}_T)$ as a mixed hybrid-form strategy---a probability distribution over $\mathcal{F}_T$. Note that the number of pure hybrid-form strategies ($|\mathcal{F}_T|$) is infinite because each strategy in $\mathcal{R}_1$ corresponds to at least one pure hybrid-form strategy, and the size of $\mathcal{R}_1$ is infinite. 

Given a strategy profile $(f_T,\sigma_n)$, we define an extended utility function of the team's pure strategy similarly as $U_T(\pi_T,\sigma_n)$ in Eq.(\ref{U_T_X_sigma}): $U_T(f_T,\sigma_n)=\sum_{l\in L_{f_T,\sigma_n}}r_{1}(\text{seq}_1(l))u_T(l)c(l)$, with $L_{f_T,\sigma_n}$ being the leafs reachable by the profile, as shown in Example \ref{TMECorExampleNotation2}. The expected utility of a mixed strategy is then: 
\begin{align*}\textstyle U_T(\overline{x}_T,\sigma_n)=\sum_{f_T\in \mathcal{F}_T}U_T(f_T,\sigma_n)\overline{x}_T(f_T).\end{align*}
 Using a sequence-form strategy $r_n$ of the adversary, an extended utility function of the team's pure strategy is defined (similarly to $U_T(\pi_T,r_n)$ as in Eq.(\ref{U_T_X_r})) as $U_T(f_T,r_n)=\sum_{l\in L_{f_T,r_n}}r_{n}(\text{seq}_n(l))r_{1}(\text{seq}_1(l))u_T(l)c(l)$, with $L_{f_T,r_n}$ defined accordingly, as shown in Example \ref{TMECorExampleNotation2}. The corresponding expected utility of a mixed strategy is: 
 $$\textstyle U_T(\overline{x}_T,r_n)=\sum_{f_T\in \mathcal{F}_T}U_T(f_T,r_n)\overline{x}_T(f_T).$$    
\begin{Example}\label{TMECorExampleNotation2}
On the right side of Figure \ref{TMECorNotationExample} we depict a sequence-form strategy of player 1 $r_1$ with $r_1(a)=0.2$ and $r_1(b)=0.8$, and a normal-form strategy of player 2 $\pi_2=de$. The corresponding hybrid-form strategy is then $f_T=(r_1,\pi_2)=((0.2,0.8),de)$. $r_1$ is equivalent to a mixed normal-form strategy $x_1$ with $x_1(\pi_1=a)=0.2$ and $x_1(\pi'_1=b)=0.8$, where $\pi_1=a$ and $\pi'_1=b$ are pure normal-form strategies. The strategy $f_T$ can be represented as a joint strategy $(x_1=(0.2,0.8),\pi_2=de)$, which is equivalent to a mixed joint normal-form strategy $x_T$ with $x_T(\pi_1,\pi_2)=0.2$ and $x_T(\pi'_1,\pi_2)=0.8$. 
In this case, the mixed joint normal-form strategy $x_T$ corresponds to a mixed hybrid-form strategy $\overline{x}_T$ with $\overline{x}_T(f_T)=1$. 
Given the adversary's sequence $\sigma_3=j$, the set of terminal nodes reachable by $(f_T,\sigma_3)$ is denoted as $L_{f,\sigma_3}=\{3\}$. If $\sigma_3=m$, the set $L_{f_T,\sigma_3}=\{5\}$. For an arbitrary   sequence-form strategy $r_3$ of the adversary, we have $L_{f_T,r_3}=\{3,4,5,6\}$.
\end{Example}
We define the support sets in a mixed normal-form strategy and a mixed hybrid-form strategy as:
\begin{align*}&S_{x_T}=\{\pi_T\mid x_T(\pi_T)>0, \pi_T\in \Pi_T\},\\ &S_{\overline{x}_T}=\{f_T\mid \overline{x}_T(f_T)>0, f_T\in \mathcal{F}_T\}.\end{align*}
 
The first step in computing a TMECor using hybrid-form strategies is to show that the set of TMEsCor is preserved under this strategy representation. In other words, the utility of any player (the team or the adversary) in each TMECor with a hybrid-form strategy has to be the same as in a TMECor with a normal-form strategy of the team. Because we consider zero-sum games, it is enough to show that both strategy representations guarantee the same utility of the team. To prove this result, we only need to consider each sequence of the adversary once because a sequence-form strategy of the adversary is defined by the probability for taking each sequence.
We first prove a more general result: for any mixed normal-form strategy of the team, there exists a mixed hybrid-form strategy such that the team obtains the same  expected utility in both strategy representations with any strategy of the adversary, and vice versa. 
Recall that two strategies of the team are realization-equivalent if they both define the same probabilities for reaching nodes given any strategy of the adversary. Based on this definition, in the following two lemmas, we show the realization equivalence between the hybrid-form and the normal-form strategies, as shown in Example \ref{TMECorExampleNotation2}. We also show that two realization-equivalent strategies could have the same size of support sets in Lemma \ref{realizationEQLemma}.  All proofs can be found in Appendix \ref{TMECorProofs}.
\begin{Lemma}\label{realizationEQLemma}
For every normal-form strategy $x_T\in\Delta(\Pi_T)$ there exists a realization-equivalent hybrid-form strategy $\overline{x}_T\in\Delta(\mathcal{F}_T)$ such that $|S_{x_T}| = |S_{\overline{x}_T}|$.
\end{Lemma}

\begin{Lemma}\label{realizationEQ}
For every hybrid-form strategy $\overline{x}_T\in\Delta(\mathcal{F}_T)$ there exists a realization-equivalent normal-form strategy $x_T\in\Delta(\Pi_T)$.
\end{Lemma}
Using the two lemmas, we can prove that our hybrid-form strategies preserve the set of TMEsCor\footnote{Note that our proofs do not rely on the inexecutable realization-form strategies of \citeauthor{farina2018ex} \citeyearpar{farina2018ex}.}.
\begin{Theorem}\label{coroTMECorEQU}
For every normal-form strategy $x_T\in \Delta(\Pi_T)$ there exists a hybrid-form strategy $\overline{x}_T\in\Delta(\mathcal{F}_T)$ such that $U_T(x_T,\sigma_n)=U_T(\overline{x}_T,\sigma_n)$, $\forall \sigma_n\in\Sigma_n$, and vice versa.
\end{Theorem}

Because the set of TMEsCor will not change if hybrid-form strategies are played, the following linear program (LP)   computes the equilibrium:
  \begin{subequations}\label{TMECorprogramNew}
\begin{align}
&\textstyle \max_{\overline{x}_T}  v(\mathcal{I}_n(\varnothing))\\
& \quad\quad \textstyle v(\mathcal{I}_n(\sigma_n)) -\sum_{I_{n,j}\in I_n:\text{seq}_n(I_{n,j})=\sigma_n} v(I_{n,j}) \nonumber\\
&\quad \quad\quad\quad\quad\quad \quad\quad\quad \leq U_T(\overline{x}_T,\sigma_n)\quad \forall \sigma_n\in \Sigma_n \label{TMECorprogramNewEq2}\\
&\quad\quad\textstyle\sum_{f_T\in \mathcal{F}_T}\overline{x}_T(f_T)=1\\
&\quad\quad \overline{x}_T(f_T)\geq 0 \quad \forall f_T\in \mathcal{F}_T.
\end{align}
\end{subequations}
In this LP, $\mathcal{I}_n(\sigma_n)$ denotes an information set in which   player $n$ takes the last action of sequence $\sigma_n$, and $v(I_{n,j})$ is the expected utility of the team in each information set $I_{n,j}$. The adversary chooses the strategy minimizing the team's utility in each information set $\mathcal{I}_n(\sigma_n)$ (represented by Eq.\eqref{TMECorprogramNewEq2}   \cite{bosansky2014exact}), and $v(\mathcal{I}_n(\varnothing))$ is hence the team's utility given $\overline{x}_T$.

\begin{algorithm}[tt]
\small
\caption{CG Algorithm with a Multilinear BRO}\label{GC}
Initialize $\mathcal{F}'_T\leftarrow\{\text{any strategy}\},\underline{v}\leftarrow0,\overline{v}\leftarrow1$\label{CMB_line1}\;
\Repeat {$\underline{v} = \overline{v} $\label{CMB_line2}}{
    $(\underline{v},\overline{x}_T,r_n) \leftarrow$    CoreLP$(\mathcal{F}'_T,\Sigma_n)$\label{CMB_line3}\;  
    $(\overline{v},f_T)\leftarrow $ 
{BRO}$(r_n)$\label{CMB_line4}\;
$\mathcal{F}'_T\leftarrow \mathcal{F}'_T\cup\{f_T\}$\label{CMB_line5}\;
}
\Return $(\overline{x}_T,r_n)$.\label{CMB_line7}
\end{algorithm}

\subsection{Column Generation with a Multilinear Oracle}\label{sectionCMB} 
Although we can formulate the problem of finding a TMECor as LP (\ref{TMECorprogramNew}), solving it requires enumerating all pure hybrid-form strategies in an infinite strategy space as variables, which is impractical. To address this problem, we introduce a CG algorithm with a Multilinear BRO (CMB), depicted in Algorithm \ref{GC}. Our CMB starts from a restricted game $(\mathcal{F}'_T,\Sigma_n)$ ($\mathcal{F}'_T$ is a subset of $\mathcal{F}_T$) and proceeds iteratively. It computes a TMECor in a subgame and checks if there exists a better strategy outside the restricted game. If the answer is positive, it expands the restricted game, and terminates otherwise.

Now we describe this dynamic in more detail. In Algorithm \ref{GC},   CoreLP$(\mathcal{F}'_T,\Sigma_n)$ in Line \ref{CMB_line3} computes a TMECor $(\overline{x}_T,r_n)$ with the team's utility $\underline{v}$ in the restricted game $(\mathcal{F}'_T,\Sigma_n)$ by solving LP \eqref{TMECorprogramNew}. On the next line, the output $f_T$ of our BRO$(r_n)$ is a pure hybrid-form strategy maximizing the team's utility $\overline{v}$ when the adversary plays $r_n$, i.e., $\arg\max_{f_T\in \mathcal{F}_T}U_T(f_T,r_n)$. We will show how to compute it efficiently later.
Thus, starting from the restricted $\mathcal{F}'_T  $ (initialized at Line \ref{CMB_line1}), CMB computes the equilibrium $(\overline{x}_T,r_n)$ with the team's utility $\underline{v}$ in the restricted game $(\mathcal{F}'_T,\Sigma_n)$ by calling CoreLP    (Line \ref{CMB_line3}),  and finds the team's best response $f_T$ with the team's utility $\overline{v}$ against the adversary's strategy $r_n$ in the equilibrium of $(\mathcal{F}'_T,\Sigma_n)$  by calling  BRO (Line \ref{CMB_line4}). Then CMB expands $\mathcal{F}'_T$ if a new strategy $f_T$ outside of $\mathcal{F}'_T$ is found (Line \ref{CMB_line5}), otherwise CMB terminates (Line \ref{CMB_line7}).

We aim to show that once the algorithm terminates, the equilibrium in the restricted game is an equilibrium in the full game $(\mathcal{F}_T,\Sigma_n)$. Unfortunately, most column-generation algorithms usually rely on the fact that the strategy space is finite, and the number of iterations is hence bounded to enforce a guarantee of convergence~\cite{bosansky2014exact,mcmahan2003planning}. Since a strategy space in our CMB is infinite, this argument can not hold. In the following proposition, we show that despite the number of pure hybrid-form strategies being infinite, the CMB will always terminate in a finite number of steps. The intuition behind our result is that if a best-response strategy $f_T$ is added to the restricted game, instead of representing a single strategy, it stands for a whole subset of hybrid-form strategies in the BRO due to the mixed sequence-form strategy of the first player, and the number of these subsets of hybrid-form strategies is finite. 
\begin{Proposition}\label{upperboundofiteration}
CMB converges in at most $2^{|\Pi_{1}|}\frac{|\Pi_T|}{|\Pi_1|}$ steps.
\end{Proposition}

Even though this result implies that there always exists a TMECor with a finite-sized support set for the team (i.e., at most $2^{|\Pi_{1}|}\frac{|\Pi_T|}{|\Pi_1|}$ strategies in the support set), it is a rather weak guarantee. We hence provide a tighter bound on the size of the team's support set\footnote{A similar result was proven for normal-form strategies earlier~\cite{celli2018computational}; however, Proposition \ref{supportUpperBound} is not its direct consequence as it relies on Lemma \ref{realizationEQLemma} that makes the connection between normal-form and hybrid-form strategies.}. That is, there exists a TMECor such that  the size of the support set of the hybrid-form strategy will not be larger than the number of adversary's sequences, which is significantly smaller than  $2^{|\Pi_{1}|}\frac{|\Pi_T|}{|\Pi_1|}$.      
\begin{Proposition}\label{supportUpperBound}
There is a TMECor $\overline{x}_T$ with $|S_{\overline{x}_T}|\leq |\Sigma_n|$. 
\end{Proposition}

Our experimental results show that the actual number of iterations before CMB converges is significantly smaller than the bound from Proposition~\ref{upperboundofiteration}. We attribute this to the fact that at least one TMECor has a small support set, as shown in Proposition~\ref{supportUpperBound}. Finally, we prove that the output of CMB is a TMECor in the full game because the team cannot find a better strategy outside the restricted game. 
That is, the output of CMB is a TMECor in the restricted game $(\mathcal{F}'_T,\Sigma_n)$ and also   a  TMECor in the full game $(\mathcal{F}_T,\Sigma_n)$.
\begin{Theorem}
  CMB converges to a TMECor in $(\mathcal{F}_T,\Sigma_n)$.
\end{Theorem}

Moreover, we show that if the BRO checks for the existence of a best-response strategy with a gain at most $\epsilon$ outside the restricted game, the CMB converges to an approximate TMECor. This is achieved simply by changing the termination condition from    $\overline{v}=\underline{v}$ to $\overline{v}-\underline{v}\leq \epsilon$.
In this case, the output of CMB is a TMECor in the restricted game $(\mathcal{F}'_T,\Sigma_n)$ and also   an approximate TMECor in the full game $(\mathcal{F}_T,\Sigma_n)$.
\begin{Proposition}
If CMB terminates with   $\overline{v}-\underline{v}\leq \epsilon$, then its output $(\overline{x}_T,r_n)$ is an $\epsilon$-TMECor in $(\mathcal{F}_T,\Sigma_n)$.
\end{Proposition}
Now we introduce our BRO. Note that:
\begin{align*}
    U_T(f_T,r_n)&\textstyle \!=\!\sum_{l\in L_{f_T,r_n}}\!r_{n}(\text{seq}_n(l))r_{1}(\text{seq}_1(l))u_T(l)c(l)\\
    &\textstyle \!=\!\sum_{l\in L} \!u_{T}(l)c(l)r_n(\text{seq}_n (l) ) \prod_{i\in T} r_i(\text{seq}_i(l) ).
\end{align*}
Then the BRO can be formulated as the following multilinear program that computes a best response $f_T$ against a sequence-form strategy  $r_n$ of the adversary, i.e., $\arg\max_{f_T\in \mathcal{F}_T}U_T(f_T,r_n)$:
\begin{subequations}\label{TMECorBRO}
\begin{align}
&  \max_{\times_{ i\in T} r_{i}} \textstyle  \sum_{l\in L} u_{T}(l)c(l)r_n(\text{seq}_n (l) ) \prod_{i\in T} r_i(\text{seq}_i(l) )\\
&\quad\quad\text{Eqs.}(\ref{sqconstraint1ad})-(\ref{sqconstraint2ad})\quad \forall i\in T\label{TMECorBROcon1}\\
&\quad\quad r_i(\sigma_i)\in\{0,1\} \quad \forall  \sigma_i\in \Sigma_i,i\in T\setminus\{1\}\\
&\quad\quad r_{1}(\sigma_{1})\in [0,1] \quad  \forall \sigma_{1}\in \Sigma_{1},\label{TMECorBROcon3}
\end{align}
\end{subequations}
where $r_i\in \overline{\mathcal{R}}_i$ is realization-equivalent to $\pi_i\in\Pi_i$ for all $i$ in $ T\setminus\{1\}$ in a hybrid-form strategy $f_T$.
The idea of the BRO is that it expresses the probability of all team players reaching each terminal node using a multilinear term $\prod_{i\in T}r_i(\text{seq}_i(l))$  in $U_T(f_T,r_n)$ with only $\sum_{i\in T\setminus\{1\}}|\Sigma_i|$ integer variables.

\subsection{Associated Representation Technique}
\label{sectionART}
Because problem \eqref{TMECorBRO} is multilinear and thus difficult to solve, we develop a novel global optimization technique for finding a solution efficiently. We call the method the Associated Representation Technique (ART). The ART represents the multilinear terms exactly through linear constraints. Moreover, it reduces the feasible solution space by using associated constraints for the equivalence relations between the individual multilinear terms. From the computational perspective, ART's two essential properties are that (i) it does not require recursive expansion to represent the multilinear terms exactly, and (ii) it generates the associated constraints efficiently.


\subsubsection{Multilinear Representation (MR)}
First, we show how to transform problem \eqref{TMECorBRO} into an equivalent MILP exactly,  without introducing new integer variables. For this purpose, for each multilinear term $w(\sigma_T)= \prod_{ i\in T }r_i(\sigma_T(i))$, where  $\sigma_T(i)$ is the sequence of player $i$ in joint sequence $\sigma_T\in \Sigma_T $ with $\Sigma_T =\times_{i\in T}\Sigma_i$,  $r_i\in \overline{\mathcal{R}}_i $ for all $i\in T\setminus\{1\}$, and $r_1\in \mathcal{R}_1$, we introduce the following MR constraints:
\begin{subequations}
\begin{align}
&0\leq w(\sigma_T)\leq r_i(\sigma_T(i)) \quad \forall  i\in T\setminus\{1\} \label{werconstraint1}\\
&0\leq r_{1}(\sigma_T( 1) ) - w(\sigma_T) \leq n-2 - \!\!\!\!\! \sum_{i\in T\setminus\{ 1 \}} \!\!\!\!\!  r_i(\sigma_T(i)  ). \label{werconstraint2}
\end{align}
\end{subequations}

Note that the multilinear term $\prod_{ i\in T }r_i(\sigma_T(i))$ is equal to the continuous variable $r_1(\sigma_T(1))$ if all binary variables are set to 1, and it is 0 if there is a binary variable with value 0. Now we show that variable $w(\sigma_T)$ in Eqs.(\ref{werconstraint1})--(\ref{werconstraint2}) exactly represents the multilinear term $\prod_{ i\in T }r_i(\sigma_T(i))$.
\begin{Proposition}\label{MultilinearRepresentation}
  $\prod_{ i\in T }r_i(\sigma_T(i))$ with $r_i\in \overline{\mathcal{R}}_i $ for all $i\in T\setminus\{1\}$ and $r_1\in \mathcal{R}_1$ is exactly represented by $w(\sigma_T)$ in Eqs.(\ref{werconstraint1}) and  (\ref{werconstraint2}).
\end{Proposition}


\subsubsection{Efficient Generation of Associated Constraints}
By using the MR constraints, problem \eqref{TMECorBRO} becomes an MILP, which can be solved using a standard branch-and-bound approach with an LP relaxation \cite{morrison2016branch}. However, relaxing the MR constraints may result in a much larger feasible solution space. To be more specific, as a consequence of making the variable $r_i(\sigma_T(i))$ real-valued, the variable $w(\sigma_T)$ may no longer exactly represent $\prod_{ i\in T }r_i(\sigma_T(i))$ in  Eqs.(\ref{werconstraint1}) and (\ref{werconstraint2}), as intended. Therefore, we aim to reduce the feasible solution space of $w(\sigma_T)$. For this purpose, we generate associated constraints enforcing equivalence relations between multilinear terms, that are based on network-flow constraints (\ref{sqconstraint1ad})-(\ref{sqconstraint2ad}) for the sequence-form strategies. As an example, suppose that we have the multilinear terms $w,w_1,w_2$ and $w'$ with $w=r_1(\sigma_1)w', w_1=r_1(\sigma_1a)w', w_2=r_1(\sigma_1b)w'$. A constraint for the sequence-form strategy $r_1$ requires that $r_1(\sigma_1)= r_1(\sigma_1a)+ r_1(\sigma_1b)$. Therefore, we can introduce an associated constraint $w=w_1+w_2$. Adding associated constraints immediately rules out some candidate solutions, which effectively reduces the MILP's solution space and results in faster computation.

\begin{table*}
  \centering
   \scalebox{1.0}{
   \begin{tabular}{l|ll|lllllll}
     \hline
       EFG   &   $|L|$  & $|\Sigma_i|$ &    CMB  & CMB/H    &CMB/A   &    CMB/ART &CMB/ART/H  & C18 & F18 \\\hline
      3K4&312&33&{\bf 0.7s} &0.7s&2.1s&4s& 6s  &6.8s&1.2s\\
      3K6&1560&49&{\bf2s}&2s &20s&191s&479s  &$>$5h&12s\\
      3K8&4368&65&{\bf4s}&4s &497s&7160s& $>$5h &&210s\\
      3K10&9360&81&{\bf 5s}& 6s &10530s&$>$5h&  &&3541s \\
      3K12&17160&97&{\bf 10s}&10s &$>$5h&&&&$>$5h\\
            4L3$_1$&30600&219&{\bf68s}&165s&&&&&\\
                  4L3$_2$& 638064 &219&{\bf1264s}&2155s&&&&&\\
                        3L3&249480&457&{\bf 4916s}&6500s&&&&&\\
                        4K9&99792&145&{\bf2.6h}&3.8h&&&&&\\
      3L5$^*$&10020&1001&{\bf4.4h}&$>$6h&&&&&\\

      \hline
   \end{tabular}
   }
      \caption{The runtimes of algorithms computing TMECor. The difficulty of finding a solution increases from top to bottom. We use the notation `$>n$h' to indicate that an algorithm did not terminate after $n$ hours on the current and all larger instances. We assume that the largest 3L5$^*$ instance has five cards, and team players do not take action ``raising'' in 4L3$_1$ (6 cards) and 4L3$_2$. }\label{kuhnLeducpokerTMECorResults}
\end{table*}

\begin{table}
  \centering
   \begin{tabular}{p{12mm}|p{6mm}p{6mm}p{6mm}p{6mm}p{5mm}p{5mm}p{5.5mm} }
     \hline
       EFG      &   5K11  &    5K12 &5K13 & 6K7&6K8& 6K9&7K7 \\\hline
      CMB   &{\bf35s}&{\bf58s}&{\bf104s}&{\bf17s}&{\bf64s}&{\bf 216s}&{\bf56s}\\
      CMBZ20  &2802s&6319s&$>$3h&2009s&$>$3h&$>$3h &$>$15h\\
      \hline
   \end{tabular}
       \caption{The runtimes of the CMB algorithm and the CMBZ20 algorithm for computing TMECor in larger games. We use the same notation as in Table~\ref{kuhnLeducpokerTMECorResults}.
    Here, team players  choose actions in information sets reaching by sequence $\varnothing$ and then take action 'calling' in other information sets.
       }\label{GACTMECorResults}
\end{table}
Because the associated constraints are closely related to the network-flow constraints of sequence-form strategies, they can be generated in a similar manner---through information sets. And because the variable  $w(\times_{i\in T}\text{seq}_i(l))$ used by MR constraints to represent the product $\prod_{ i\in T }r_i(\text{seq}_i(l))$ involves all team members' sequences, the associated constraints are generated for all information sets of all members. 
For example, consider a four-player Kuhn poker game, in which a terminal node is reached by
the team's joint sequence tuple $(\text{\it J:/cccr:c, Q:/cccrc:c, T:/cccrcc:c})$. The three sequences are taken by three team players in information sets $\text{\it J:/cccr:, Q:/cccrc:, T:/cccrcc:}$, respectively. Assume that the information set $\text{\it  T:/cccrcc:}$ of player 3 is reachable by a sequence $\text{\it T:/cc:c}$. The information set contains a node $(\text{\it J:/cccr:c, Q:/cccrc:c, {T:/cc:c},} \text{\it K:/ccc:r})$, specified by one sequence per each player. There are two possible actions that can be taken: action $c$ and action $f$. The network-flow constraint associated with this information set is hence
\begin{align*} r_3(\text{\it T:/cc:c})=r_3(\text{\it T:/cccrcc:c})+r_3(\text{\it T:/cccrcc:f}).
\end{align*} 
The corresponding associated constraint is
\begin{align*}
&w(\text{\it J:/cccr:c, Q:/cccrc:c, {T:/cc:c}})\\
=&w(\text{\it J:/cccr:c, Q:/cccrc:c, {T:/cccrcc:c}}) \\
&+w(\text{\it J:/cccr:c, Q:/cccrc:c, {T:/cccrcc:f}}).
\end{align*}
Now assume that there is another node $(\text{\it J:/cccr:c, K:/cccrc:c,} $ $\text{\it{T:/cc:c}, Q:/ccc:r})$ in the same information set, and a terminal node with the team's joint sequence $(\text{\it J:/cccr:c,}$ $\text{\it K:/cccrc:c, T:/cccrcc:c})$ reachable by action $c$ of player 3. The following associated constraint is generated:
\begin{align*}&w(\text{\it J:/cccr:c, K:/cccrc:c, {T:/cc:c}})\\ 
=& w(\text{\it J:/cccr:c, K:/cccrc:c, {T:/cccrcc:c}}) \\
&+ w(\text{\it J:/cccr:c, K:/cccrc:c,  {T:/cccrcc:f}}).
\end{align*}
Using the same approach, we can generate associated constraints in this Kuhn poker game in all team players' information sets. More details can be found in Appendix \ref{Example4pKuhnInf}.

Therefore, in a general EFG, in each information set $I_{i,j}$ of a team member $i$, the algorithm for generating associated constraints needs to enumerate all the team's joint sequences leading to $I_{i,j}$, which correspond to different nodes in $I_{i,j}$. We denote this set of sequences as $\Sigma_T(I_{i,j})$ and use seq$_i(I_{i,j})=\sigma_T(i)$ for all $\sigma_T=(\sigma_T(i),\sigma_{T\setminus \{i\}})\in \Sigma_T(I_{i,j})$. The associated constraints for $I_{i,j}$ are then
\begin{equation}
\begin{aligned}
\label{associatedConstraintEQ}
&\!\!\!\! \textstyle w(\sigma_T)\!=\!\sum_{a\in \chi(I_{i,j})} \!w(\sigma_T(i)a,\sigma_{T\setminus \{i\}}) \\
&\quad \quad\quad\quad \quad \quad \forall \sigma_T\in \Sigma_T(I_{i,j}), I_{i,j}\in I_i,i\in T.
\end{aligned}
\end{equation}
Generating all the constraints can be thus done in time $O(\sum_{i\in T}\sum_{I_{i,j}\in I_i}|\Sigma_T(I_{i,j})|)$. The resulting MILP for representing problem (\ref{TMECorBRO}) using Eqs.(\ref{werconstraint1})-(\ref{werconstraint2}) and (\ref{associatedConstraintEQ}) can be formulated as follows:
\begin{subequations}\label{TMECorBROART}
\begin{align}
&  \max_{\times_{\!i\in T} r_{i}} \! \sum_{l\in L}u_{T}(l)c(l)r_n(\text{seq}_n\!(l)\!)w(\times_{i\in T}\text{seq}_i(l))\\
&\quad\quad\text{Eqs.}(\ref{TMECorBROcon1})-(\ref{TMECorBROcon3}), (\ref{associatedConstraintEQ}) \\
&\quad\quad\text{Eqs.} (\ref{werconstraint1})-(\ref{werconstraint2})\quad \forall w(\times_{i\in T}\text{seq}_i(l)), l\in L.
\end{align}
\end{subequations}
Our final theorem proves that associated constraints preserve the sequence-form strategy space, making the solution of formulation (\ref{TMECorBROART}) also a feasible solution of our BRO in CMB. The intuition is that associated constraints are consistent with the sequence-form constraints and hence do not alter the space of feasible sequence-form solutions in Problem (\ref{TMECorBRO}).  
\begin{Theorem}\label{theoremFeasibleART}
The solution of Problem (\ref{TMECorBRO}) solves Problem (\ref{TMECorBROART}). 
\end{Theorem}

Proposition \ref{MultilinearRepresentation} guarantees that solutions of both problems will share the same value. The result, however, is even stronger: the   optimal solution of Problem (\ref{TMECorBROART}) is also optimal for Problem (\ref{TMECorBRO}). Thus, it is a best response against the adversary's strategy $r_n$. 
\begin{Corollary} \label{theoremApproachART} For any strategy $r_n$ of the adversary, the optimal solution of Problem (\ref{TMECorBROART}) is a best response against $r_n$.   
\end{Corollary}

\begin{table}[t]
  \centering
   \begin{tabular}{l|p{5mm}p{5mm}llll}
     \hline
       EFG      &   3K4  &    3K6 &3K8 & 3K10&3K12& 3L3 \\\hline
      Iterations   &14&36&47&45&52&1022\\
      Support size  &3&5&8&6&10&63\\
      \hline
   \end{tabular}
       \caption{The number of iterations until CMB converges, and the size of the support set of the team's TMECor strategy.
       }\label{TMECorResultsSupportSizeIteration}
\end{table}

\begin{table}[t]
  \centering
   \begin{tabular}{p{6mm}|p{6mm}p{6mm}p{6mm}p{6mm}p{6mm}p{6mm}p{6mm} }
     \hline
       $\epsilon$   &   0.1  & 0.08 &    0.06  &    0.04  &    0.02  & 0.01&0.008  \\\hline
      CMB&{\bf0.41s} &{\bf 0.41s} &{\bf0.41s} &{\bf0.50s} &{\bf0.50s} &{\bf0.58s} &{\bf0.58s}\\
      FTP& 0.55s& 0.71s&0.82s &1.3s &3.8s&8.0s&11.0s\\\hline
      CMB& {\bf 0.08s} &{\bf0.08s} &{\bf0.16s} &{\bf0.16s} &{\bf0.29s} &{\bf0.50s} &{\bf0.87s}\\
      FTP&4.0s &6.6s &8.1s & 15.5s&72.3s&181s&243s\\\hline
      CMB& {\bf0.23s} &{\bf0.23s} &{\bf0.31s} &{\bf0.31s} &{\bf0.64s} &{\bf1.1s} &{\bf1.1s}\\
      FTP&34.6s &34.6s &67.7s &94.9s &171s&382s&458s\\\hline
      CMB&  {\bf2s} &{\bf 3s} &{\bf6s } &{\bf 13s} &{\bf41s} &{\bf87s} &{\bf111s}\\
      FTP&228s &307s &458s &689s &1574s&4882s&$>$5h\\\hline
      CMB &{\bf 5s} &{\bf 7s} &{\bf 13s} &{\bf 28s} &{\bf93s} &{\bf745s} &{\bf1533s}\\
      FTP&188s &251s & 362s& 619s&$>$5h&&\\\hline
      CMB  &{\bf23s} &{\bf27s} &{\bf37s} &{\bf108s} &{\bf490s} &{\bf3357s} &{\bf8221s}\\
      FTP&164s &   215s&371s &920s&$>$5h&&\\\hline
   \end{tabular}
       \caption{The runtimes of algorithms computing $\epsilon\Delta_u$-TMEsCor. The games from the top to the bottom are 3K4, 3K8, 3K12, 3L3, 3L4, and 3L5.   $\Delta_u=6$ for $3Kr$ and $\Delta_u=21$ for $3Lr$. The sizes and  notations are the same as in Table~\ref{kuhnLeducpokerTMECorResults}. In addition,   $|L|\approx 10^6$ with $|\Sigma_i|=801$ for 3L4, and $|L|\approx 3\cdot10^6$ with $|\Sigma_i|=1241$ for 3L5. 
       }\label{kuhnLeducpokerEpsilonTMECorResults}
\end{table}

\begin{table}[t]
  \centering
   \begin{tabular}{ lp{11mm}|p{8.5mm}p{8.5mm}p{8.5mm}p{8.5mm}p{8.5mm} }
     \hline
     \multicolumn{2}{c|}{   EFG}       &   3K8  &    3K9 &3K10 & 3K11&3K12 \\\hline
   \multicolumn{1}{c|}{\multirow{2}{*}{$t$  }}  & TMECor    &{\bf4s}&{\bf4s}&{\bf5s}&{\bf10s}&{\bf10s}\\
     \multicolumn{1}{c|}{}& TME &4s&5s&84s&437s&118s\\\hline
     \multicolumn{1}{c|}{\multirow{3}{*}{$u$}}  & TMECor   &{\bf-0.019}&{\bf-0.018}&{\bf-0.016}&{\bf-0.015}&{\bf-0.014}\\
      \multicolumn{1}{c|}{}& TME  &-0.066&-0.044&-0.068&-0.050&-0.055\\
      \multicolumn{1}{c|}{}& Gap  &71\%&59\%&77\%&71\%&74\%\\
      \hline
   \end{tabular}
       \caption{The runtimes $t$ and the team's utilities $u$ for computed TMEs and TMEsCor solutions. We calculate the gap as the relative distance between the team's utility ($u_{Cor}$)  in a TMECor  and the one ($u_{TME}$) in a TME, i.e., $\frac{|u_{TME}-u_{Cor}|}{|u_{TME}|}\times 100\%$. Greater gap indicates that the team will lose more if it opts for the TME strategy. 
       }\label{TMECorResultsvsTME}
\end{table}

\section{Experimental Evaluation}

Finally, we demonstrate the performance of our CMB algorithm. We compare CMB to the previous state-of-the-art algorithms: (i) the original CG algorithm in Celli and Gatti \citeyearpar{celli2018computational} (referred to as C18); (ii) the CG with the BRO proposed by Farina et al. \citeyearpar{farina2018ex} (referred to as F18), and  (iii) the CMB with the associated constraints generation algorithm of Zhang and An \citeyearpar{zhang2020computing} (referred to as CMBZ20). We use two standard EFG domains for evaluating the algorithms: (i) the Kuhn poker, and (ii) the Leduc Hold'em poker.  Formal definitions of the domains can be found in Appendix~\ref{pokergamerules}. All       players have the same number of sequences in these games, and we use   $|\Sigma_i|$ to represent this number of sequences. We denote an $n$-player Kuhn instance with $r$ ranks (i.e., $r$ cards) as $n$K$r$, and refer to an $n$-player Leduc Hold'em instance with $r$ ranks (i.e., $3r$ cards) as $n$L$r$.  Without a loss of generality, the last player assumes the role of the adversary. All (MI)LPs are solved by CPLEX 12.9 on a machine with 6-core 3.6GHz CPU and 32GB of memory.

{\bf Runtimes.} We present the runtime results in Table~\ref{kuhnLeducpokerTMECorResults}. We omit the runtimes of CMBZ20 because it performs similarly to CMB, but evaluate their differences further in larger games and report the results in Table~\ref{GACTMECorResults}. For assessing ablations, we compare CMB to its four variant: (i)   CMB without associated constraints (referred to as CMB/A);  (ii)   CMB that uses continuous variables to represent reaching probabilities for terminal nodes without using our ART and BRO (referred to as CMB/ART); (iii)   CMB with BRO generating  joint pure normal-form strategies instead of hybrid-form strategies (referred to as CMB/H), and (iv) CMB/ART/H---a combination of (ii) and (iii). The results clearly show that CMB is several orders of magnitude faster than the reference algorithms. Moreover, CMB also outperforms all ablation algorithms, which strongly suggests that each component of CMB significantly boosts its performance.

{\bf Convergence and Supports.} In Table~\ref{TMECorResultsSupportSizeIteration}, we report the number of iterations CMB needs to converge to an exact TMECor, together with the team's equilibrium strategy support size. The number of iterations is significantly smaller than the theoretical upper bound $2^{33}\times 33$ derived in Proposition~\ref{upperboundofiteration}\footnote{$|\Sigma_i| = 33$ in 3K4, whereas $|\Pi_i|$ is significantly greater.}. The support sets in TMECor also remain small.

{\bf Approximation.} In the next experiment, we evaluate CMB's ability to compute an $\epsilon\Delta_u$-TMEsCor, where $\Delta_u$ is the difference between the maximum and minimum achievable utility of the team. We compare CMB to Fictitious Team Play (FTP) \cite{farina2018ex}. We use the setting of FTP reported in Farina  et al. \citeyearpar{farina2018ex}, including their BRO's time limit of 15s. Note that their BRO cannot run on large games otherwise. The results in Table \ref{kuhnLeducpokerEpsilonTMECorResults} show that CMB runs significantly faster than FTP. For example, CMB is at least two orders of magnitude faster than FTP on large games with small $\epsilon$, e.g., $\epsilon=0.01$. Moreover, according to the results, it is almost impossible for FTP to converge to an  $\epsilon\Delta_u$-TMECor with an even smaller $\epsilon$ (e.g., $\epsilon=0.0001$), let alone an exact TMECor. For 3K4---the smallest game in our experiments---FTP is unable to converge to an  $\epsilon\Delta_u$-TMECor with $\epsilon=0.0004$ in 100 hours\footnote{$\epsilon$ reaches $0.0005$ in 2674s but then it fluctuates around $0.001$.}. In contrast, CMB computes an exact TMECor within 0.7s, as shown in Table \ref{kuhnLeducpokerTMECorResults}.

{\bf Comparison to TME.} The last experiment demonstrates the difference between TMECor and TME discussed in Section \ref{sectionTMECorRelatedwork}, both in runtime and the team's utility. Previous literature has shown that the team suffers large losses in utility when resorting to TME strategies instead of TMECor strategies in 3K3--3K7 \cite{farina2018ex}. In Table~\ref{TMECorResultsvsTME}, we report the results on larger 3K8--3K12. Because of the encountered difficulty to compute an exact TME, we use the state-of-the-art algorithm of Zhang and An \citeyearpar{zhang2020computing} to approximate the TME by computing an $\epsilon\Delta_u$-TME with $\epsilon=0.01$. The results show that approximating a TME takes significantly longer than computing a TMECor, while at the same time, the TME strategies are inferior to the TMECor strategies.

\section{Conclusion and Future Work}
We propose a new algorithm (CMB) for finding a TMECor in large zero-sum multiplayer EFGs. Our algorithm is based on a novel hybrid-form strategy representation of the team, which gives rise to a column-generation method with guaranteed convergence in finite time. The heart of the algorithm is a multilinear best-response oracle that can be queried efficiently using our associated representation technique. We show that our algorithm computes a TMECor significantly faster than previous state-of-the-art baselines. 
In the future, we can  explore the  tighter theoretical upper bound for the number of iterations because CMB requires very few iterations in experiments. 
Due to the difficulty of computing the best-response oracle, we would like to explore the possibility of improving our oracle by   reinforcement learning \cite{timbers2020approximate} or  regret minimization \cite{celli2019learning}.

\section*{Broader Impact}
Game-theoretic solutions, including our algorithm, have both descriptive and prescriptive applications in suitable competitive environments, including businesses, politics, or even gambling. Finding the equilibria helps to understand people's behavior when interacting in dynamic situations and makes it easier to construct effective decisions to optimize multiagent systems. For example, the manufacturers in competitive markets may find better pricing strategies if they consider the decision-making of their competitors. While abusing the theories, e.g., by gamblers in casinos, is also feasible, the same approach also allows for identifying strategic violators by predicting their behavior. A notable drawback of traditional solution concepts like TMECor is their dependence on involved players' rational behavior. In case we suspect them to behave irrationally, we have to extend our models.

\section*{Acknowledgement}
This research is supported by  the National Research Foundation, Singapore under its AI Singapore Programme (AISG Award No: AISG-RP-2019-0013), National Satellite of Excellence in Trustworthy Software Systems (Award No: NSOE-TSS2019-01), the SIMTech-NTU Joint Laboratory on Complex System, and NTU.

\bibliography{TmeEfg}

\begin{thebibliography}{28}
\providecommand{\natexlab}[1]{#1}
\providecommand{\url}[1]{\texttt{#1}}
\providecommand{\urlprefix}{URL }
\expandafter\ifx\csname urlstyle\endcsname\relax
  \providecommand{\doi}[1]{doi:\discretionary{}{}{}#1}\else
  \providecommand{\doi}{doi:\discretionary{}{}{}\begingroup
  \urlstyle{rm}\Url}\fi

\bibitem[{Abou~Risk, Szafron et~al.(2010)}]{abou2010using}
Abou~Risk, N.; Szafron, D.; et~al. 2010.
\newblock Using counterfactual regret minimization to create competitive
  multiplayer poker agents.
\newblock In \emph{AAMAS}, 159--166.

\bibitem[{Bosansky et~al.(2014)Bosansky, Kiekintveld, Lisy, and
  Pechoucek}]{bosansky2014exact}
Bosansky, B.; Kiekintveld, C.; Lisy, V.; and Pechoucek, M. 2014.
\newblock An exact double-oracle algorithm for zero-sum extensive-form games
  with imperfect information.
\newblock \emph{Journal of Artificial Intelligence Research} 51: 829--866.

\bibitem[{Brown and Sandholm(2018)}]{brown2018superhuman}
Brown, N.; and Sandholm, T. 2018.
\newblock Superhuman {AI} for heads-up no-limit poker: {L}ibratus beats top
  professionals.
\newblock \emph{Science} 359(6374): 418--424.

\bibitem[{Brown and Sandholm(2019)}]{brown2019superhuman}
Brown, N.; and Sandholm, T. 2019.
\newblock Superhuman {AI} for multiplayer poker.
\newblock \emph{Science} 365(6456): 885--890.

\bibitem[{Cai and Daskalakis(2011)}]{cai2011minmax}
Cai, Y.; and Daskalakis, C. 2011.
\newblock On minmax theorems for multiplayer games.
\newblock In \emph{SODA}, 217--234.

\bibitem[{Celli and Gatti(2018)}]{celli2018computational}
Celli, A.; and Gatti, N. 2018.
\newblock Computational results for extensive-form adversarial team games.
\newblock In \emph{AAAI}, 965--972.

\bibitem[{Celli et~al.(2019)Celli, Marchesi, Bianchi, and
  Gatti}]{celli2019learning}
Celli, A.; Marchesi, A.; Bianchi, T.; and Gatti, N. 2019.
\newblock Learning to correlate in multi-player general-sum sequential games.
\newblock In \emph{NeurIPS}, 13055--13065.

\bibitem[{Chen and Deng(2005)}]{chen20053}
Chen, X.; and Deng, X. 2005.
\newblock 3-{N}ash is {PPAD}-complete.
\newblock In \emph{Electronic Colloquium on Computational Complexity}, volume
  134, 2--29.

\bibitem[{Conitzer and Sandholm(2006)}]{conitzer2006computing}
Conitzer, V.; and Sandholm, T. 2006.
\newblock Computing the optimal strategy to commit to.
\newblock In \emph{EC}, 82--90.

\bibitem[{Farina et~al.(2018)Farina, Celli, Gatti, and Sandholm}]{farina2018ex}
Farina, G.; Celli, A.; Gatti, N.; and Sandholm, T. 2018.
\newblock Ex ante coordination and collusion in zero-sum multi-player
  extensive-form games.
\newblock In \emph{NeurIPS}, 9638--9648.

\bibitem[{Farina et~al.(2020)Farina, Celli, Gatti, and
  Sandholm}]{farina2020faster}
Farina, G.; Celli, A.; Gatti, N.; and Sandholm, T. 2020.
\newblock Faster Algorithms for Optimal Ex-Ante Coordinated Collusive
  Strategies in Extensive-Form Zero-Sum Games.
\newblock \emph{arXiv preprint arXiv:2009.10061} .

\bibitem[{McCarthy et~al.(2016)McCarthy, Tambe, Kiekintveld, Gore, and
  Killion}]{mccarthy2016preventing}
McCarthy, S.~M.; Tambe, M.; Kiekintveld, C.; Gore, M.~L.; and Killion, A. 2016.
\newblock Preventing illegal logging: {S}imultaneous optimization of resource
  teams and tactics for security.
\newblock In \emph{AAAI}, 3880--3886.

\bibitem[{McCormick(1976)}]{mccormick1976computability}
McCormick, G.~P. 1976.
\newblock Computability of global solutions to factorable nonconvex programs:
  {P}art {I--C}onvex underestimating problems.
\newblock \emph{Mathematical Programming} 10(1): 147--175.

\bibitem[{McMahan, Gordon, and Blum(2003)}]{mcmahan2003planning}
McMahan, H.~B.; Gordon, G.~J.; and Blum, A. 2003.
\newblock Planning in the presence of cost functions controlled by an
  adversary.
\newblock In \emph{ICML}, 536--543.

\bibitem[{Morav\v{c}\'{i}k et~al.(2017)Morav\v{c}\'{i}k, Schmid, Burch,
  Lis\'{y}, Morrill, Bard, Davis, Waugh, Johanson, and Bowling}]{moravcik2017}
Morav\v{c}\'{i}k, M.; Schmid, M.; Burch, N.; Lis\'{y}, V.; Morrill, D.; Bard,
  N.; Davis, T.; Waugh, K.; Johanson, M.; and Bowling, M. 2017.
\newblock {DeepStack: Expert-level artificial intelligence in no-limit poker}.
\newblock \emph{Science} .

\bibitem[{Morrison et~al.(2016)Morrison, Jacobson, Sauppe, and
  Sewell}]{morrison2016branch}
Morrison, D.~R.; Jacobson, S.~H.; Sauppe, J.~J.; and Sewell, E.~C. 2016.
\newblock Branch-and-bound algorithms: A survey of recent advances in
  searching, branching, and pruning.
\newblock \emph{Discrete Optimization} 19: 79--102.

\bibitem[{Nash(1951)}]{nash1951non}
Nash, J. 1951.
\newblock Non-cooperative games.
\newblock \emph{Annals of Mathematics} 286--295.

\bibitem[{Russell and Norvig(2016)}]{russell2016artificial}
Russell, S.~J.; and Norvig, P. 2016.
\newblock \emph{Artificial Intelligence: {A} Modern Approach}.
\newblock Malaysia; Pearson Education Limited.

\bibitem[{Ryoo and Sahinidis(2001)}]{ryoo2001analysis}
Ryoo, H.~S.; and Sahinidis, N.~V. 2001.
\newblock Analysis of bounds for multilinear functions.
\newblock \emph{Journal of Global Optimization} 19(4): 403--424.

\bibitem[{Shoham and Leyton-Brown(2008)}]{shoham2008multiagent}
Shoham, Y.; and Leyton-Brown, K. 2008.
\newblock \emph{Multiagent Systems: {A}lgorithmic, Game-Theoretic, and Logical
  Foundations}.
\newblock Cambridge University Press.

\bibitem[{Sinha et~al.(2018)Sinha, Fang, An, Kiekintveld, and
  Tambe}]{sinha2018stackelberg}
Sinha, A.; Fang, F.; An, B.; Kiekintveld, C.; and Tambe, M. 2018.
\newblock Stackelberg Security Games: Looking Beyond a Decade of Success.
\newblock In \emph{IJCAI}, 5494--5501.

\bibitem[{Timbers et~al.(2020)Timbers, Lockhart, Schmid, Lanctot, and
  Bowling}]{timbers2020approximate}
Timbers, F.; Lockhart, E.; Schmid, M.; Lanctot, M.; and Bowling, M. 2020.
\newblock Approximate exploitability: Learning a best response in large games.
\newblock \emph{arXiv preprint arXiv:2004.09677} .

\bibitem[{von Stengel(1996)}]{von1996efficient}
von Stengel, B. 1996.
\newblock Efficient computation of behavior strategies.
\newblock \emph{Games and Economic Behavior} 14(2): 220--246.

\bibitem[{von Stengel and Koller(1997)}]{von1997team}
von Stengel, B.; and Koller, D. 1997.
\newblock Team-maxmin equilibria.
\newblock \emph{Games and Economic Behavior} 21(1-2): 309--321.

\bibitem[{Wichardt(2008)}]{wichardt2008existence}
Wichardt, P.~C. 2008.
\newblock Existence of Nash equilibria in finite extensive form games with
  imperfect recall: A counterexample.
\newblock \emph{Games and Economic Behavior} 63(1): 366--369.

\bibitem[{Zhang and An(2020{\natexlab{a}})}]{zhang2020computing}
Zhang, Y.; and An, B. 2020{\natexlab{a}}.
\newblock Computing team-maxmin equilibria in zero-sum multiplayer
  extensive-form games.
\newblock In \emph{AAAI}.

\bibitem[{Zhang and An(2020{\natexlab{b}})}]{zhang2020converging}
Zhang, Y.; and An, B. 2020{\natexlab{b}}.
\newblock Converging to Team-Maxmin Equilibria in Zero-Sum Multiplayer Games.
\newblock In \emph{ICML}, 11033--11043.

\bibitem[{Zinkevich et~al.(2008)Zinkevich, Johanson, Bowling, and
  Piccione}]{zinkevich2008regret}
Zinkevich, M.; Johanson, M.; Bowling, M.; and Piccione, C. 2008.
\newblock Regret minimization in games with incomplete information.
\newblock In \emph{NeurIPS}, 1729--1736.

\end{thebibliography}

\newpage \
\newpage
\section*{Appendix}
\appendix
\section{Comparison to Existing Work}\label{RelatedWrokDetail}
 A TMECor can be computed via normal-form strategies for all players.
However, in EFGs, each player's normal-form strategy space is exponential in the size of the game tree. To  compute TMEsCor more efficiently,  Celli and Gatti  \cite{celli2018computational} propose a hybrid representation of strategy spaces, in which the team plays joint normal-form strategies while the adversary plays sequence-form strategies.
Farina  et al. \citeyearpar{farina2018ex} then develop a realization-form strategy representation for all players focusing on the   probability distribution on terminal nodes of the game tree, and then they use it to derive an auxiliary game representing   the original game through a set of subtrees. Each  subtree is a two-player game between one team member and the adversary, in which the two players use realization-form strategies\footnote{Although the team member's realization-form strategy can be induced by a behavioral strategy in the subtree, the theoretical result proving the equivalence between the auxiliary game and the  original game is based solely on the realization-form strategy. Moreover, the authors' algorithm Fictitious Team-Play (FTP) also  assumes that the team member uses  a realization-form strategy that needs to be transformed into a normal-form strategy to be executed from the root. And although the team member's strategy is represented as a sequence-form strategy in the MILP formulation of the authors' BRO, it still has to be transformed into a realization-form strategy in the FTP.}, 
 while a joint pure normal-form strategy of the remaining team members is fixed.  
 Although the auxiliary game uses a hybrid-form strategy for the team, it is not a strategy representation on its own. Most importantly, it is based on the realization-form strategy that  is not executable from the root of the game tree. To make  a realization-form strategy  executable, the authors need to reconstruct an equivalent normal-form strategy. The problem is that the reconstruction is difficult in large games. The only existing reconstruction algorithm (Algorithm 2 in  Celli  et al. \citeyearpar{celli2019learning}) runs in time $O(|L|^2)$, and thus does not scale to large games (see Table \ref{kuhnLeducpokerTMECorResults} for the large number of terminal nodes). 
 To avoid such a cumbersome reconstruction algorithm, we propose a new  strategy representation for the team, which can be executed from the root immediately. That is, in our hybrid-form strategies,   one team member  uses sequence-form strategies (instead of realization-form strategies) while other team members use pure normal-form strategies.
 However, it is not clear whether our strategy representation with an infinite number of pure strategies (due to the sequence-form strategies of one member) can be used to compute exact TMEsCor (see the explanation in Section \ref{sectionCMB}).\footnote{Farina  et al. \citeyearpar{farina2018ex} do not discuss this problem for their continuous   realization-form strategy space  because they aim to compute approximate TMEsCor.} 
 Fortunately,   we theoretically show that our CG (i.e., the CMB algorithm) guarantees a convergence to a TMECor within a finite number of iterations despite our infinite strategy space. 
 Therefore, although similar structures to our hybrid-form strategy exist in the literature, we are the first to use it as the team's strategy representation without depending on the inexecutable realization-form strategies. Moreover, we theoretically show how to use it to compute exact TMEsCor.

Another key component of CG algorithms for computing TMEsCor is the BRO. The first BRO was proposed by Celli and Gatti  \cite{celli2018computational}. It expresses whether or not a leaf is reached by a pure joint normal-form strategy of all team members (see details in Appendix \ref{RelatedWorkTMECorPartCelil}) using  $|L|$ integer variables. Therefore, it is able to solve only small games.
 Farina et al. \citeyearpar{farina2018ex} develop a faster BRO that decreases a number of integer variables by expressing the utility of the joint pure normal-form strategy  of  all team members (except one) given the realization-form strategy of   other players. However, although the number of integer variables is significantly reduced to the number of sequences of  all team members (except one), this BRO is still inefficient in large games with a large number of sequences. To speed up the computation, one possible approach is to add associated constraints \cite{zhang2020computing} (see   Appendix \ref{RelatedWorkTMECorPartZhang})  to reduce the feasible solution space of the MILP representing the BRO.\footnote{Adding associated constraints, to our best knowledge, is the only known effective method to do that.} Unfortunately, the BRO of Farina et al. \citeyearpar{farina2018ex} is not compatible with associated constraints (see details in Appendix \ref{RelatedWorkTMECorPartFarina}).  Therefore, it is not possible to directly add associated constraints  to this BRO.
Overall, the main difference between  different BRO approaches is in how they represent the team's reaching probabilities for   terminal nodes with   utilities, which affects the number of integer variables and the compatibility with associated constraints.
Our formulation is the first BRO that reduces the number of integer variables and simultaneously is compatible with associated constraints (see Table   \ref{ComparingwithExiApp}).  

\begin{table*}[t]
\centering
\begin{tabular}{|c|c|c|c|c|}
  \cline{2-5}
\multicolumn{1}{c|}{ }       & \multicolumn{1}{c|}{\multirow{2}{*}{New representation}}     &  \multicolumn{3}{c|}{ BRO}     \\\cline{3-5}
\multicolumn{1}{c|}{ }   &   \multicolumn{1}{c|}{}   &  Integer variables &Compatible      & Reduce space   \\\hline
 Celli and Gatti  \citeyearpar{celli2018computational} & Yes   &$|L|$&Yes&No  \\\hline
  Farina  et al. \citeyearpar{farina2018ex} &Yes   &$\sum_{i\in T\setminus\{1\}}|\Sigma_i|$&No &No \\\hline
  Our approach & Yes &$\sum_{i\in T\setminus\{1\}}|\Sigma_i|$&Yes &Yes\\
  \hline
\end{tabular}
\caption{Comparison of existing approaches to compute TMEsCor in terms of the strategy representation, the number of integer variables ($|L|$ is significantly larger than $\sum_{i\in T\setminus\{1\}}|\Sigma_i|$ in EFGs) in BRO, the compatibility of BRO with associated constraints, and reducing the feasible solution space of BRO.}\label{ComparingwithExiApp}
\end{table*}

Associated constraints are used to reduce the feasible solution space of an MILP for computing a TME in Zhang and An \citeyearpar{zhang2020computing}. These associated constraints are generated by a global optimization technique, which is based on the recursive McCormick relaxation (see details in Appendix \ref{RelatedWorkTMECorPartMcCormick}) \cite{mccormick1976computability,ryoo2001analysis}. In our case, we can use the recursive McCormick relaxation to exactly transform the multilinear BRO into an MILP. However, we then need to recursively generate associated constraints on the path from the terminal nodes to the root, which also requires introducing new terms for sequences reaching nonterminal nodes frequently, similarly to the algorithm in Zhang and An \citeyearpar{zhang2020computing}. When the number of team players is large, the number of variables in each multilinear term increases, which results in many more performed recursive operations. This recursive generation of associated constraints undermines the effort to improve the speed   by reducing the feasible solution space when the number of team players is large. We hence develop the global optimization technique ART, which does not depend on the recursive method and generates associated constraints more efficiently. 

In the concurrent work \cite{farina2020faster}, the important step of their algorithm is to add the “probability mass conservation” constraints for the polytope in their Definition 1, which are equivalent to the associated constraints of our algorithm  in three-player games. Therefore, we expect both algorithms to perform similarly in three-player games. However, their algorithm   can only handle three-player games, while our algorithm can compute solutions even for games with four or more players by addressing additional complicated issues. More precisely, in our work we address the problem of how to efficiently generate associated constraints for games with four or more players, which is not addressed in their work. In addition, we prove the theoretical guarantee for the column generation algorithm, but they do not have such a theoretical result.

\subsection{Reaching Probability for Each Leaf   Represented by a Binary Integer Variable \cite{celli2018computational} }\label{RelatedWorkTMECorPartCelil}
The first BRO proposed for computing TMEsCor \cite{celli2018computational}  expresses whether or not a leaf is reachable by a pure joint normal-form strategy (represented by the   sequence-form strategy) of all team members.  The corresponding MILP (all payoffs of the team are assumed to be positive) is: \begin{align*} \max_{\times_{i\in T}r_i,y} & \sum_{l\in L} u_{T}(l)c(l)r_n(\text{seq}_n (l) ) y(l)  \\
&y(l)\leq r_i(\text{seq}_i(l) )\quad \forall i\in T, l\in L\\
&y(l)\in\{0,1\}\\
& \text{Eqs.}(\ref{sqconstraint1ad})-(\ref{sqconstraint2ad})\quad \forall i\in T,\label{TMECorBROcon1}
\end{align*}
where $u_T(l)$ is the utility of the team when terminal node $l$ is reached with probability $c(l)$ due to chance nodes, $y(l)\in\{0,1\}$  reflects whether or not a leaf is reached by a pure joint normal-form plan, and $r_n(\text{seq}_n(l) )$ is the known adversary's strategy.

We can see that the number of integer variables is equal to $|L|$, which is large even for games of moderate size. The large number of integer variables makes the BRO inefficient and hence slows down the CG algorithm for computing TMEsCor. This approach  can compute   TMEsCor only for very small three-player Kuhn poker games.

\subsection{Associated Constraints \cite{zhang2020computing}}\label{RelatedWorkTMECorPartZhang}
Adding associated constraints, to our best knowledge, is the only known effective method for reducing the feasible solution space of the MILP representing the BRO.\footnote{More precisely, in our case, the associated constraints reduce the feasible solution space when the MILP is solved by the branch-and-bound approach with the LP relaxation    \cite{morrison2016branch}. For simplicity, we say that we reduce the feasible space of an MILP.} Associated constraints \cite{zhang2020computing} enforce the equivalence relations between multilinear terms in accordance with   the network-flow constraints for sequence-form strategies ($\text{Eqs.}(\ref{sqconstraint1ad})-(\ref{sqconstraint2ad})$). For example, suppose we have the following multilinear terms: \begin{align*}&w=r_1(\sigma_1)w'\\ &w_1=r_1(\sigma_1a)w'\\ & w_2=r_1(\sigma_1b)w',\end{align*} and the following constraint of sequence-form strategies:  \begin{align*}r_1(\sigma_1)= r_1(\sigma_1a)+ r_1(\sigma_1b).\end{align*} Then, we can generate the associated constraint   \begin{align*}w=w_1+w_2.\end{align*}
A more concrete example is in Appendix \ref{Example4pKuhnInf}.
 \subsection{Incompatibility with Associated Constraints  \cite{farina2018ex}}\label{RelatedWorkTMECorPartFarina}
Farina  et al. \citeyearpar{farina2018ex} develop a faster BRO that expresses the utility of the joint pure normal-form strategy (represented by the joint sequence-form strategy) of  all team members (except one), given the realization-form strategy of other players. The algorithm of the authors defines a term for each joint sequence of the remaining team members. For example, in the game with three players, the utility is represented by a variable  $w(\sigma_2)$ for sequence $\sigma_2$, and then the corresponding MILP \cite{farina2018ex}, assuming all payoffs of the team are positive, is
 \begin{align*}
 \max_{w,r_1,r_2} &\sum_{\sigma_2\in \Sigma_2}w(\sigma_2)\\
 &w(\sigma_2)\leq u(\sigma_2) \quad \forall \sigma_2\in \Sigma_2\\
 &w(\sigma_2)\leq M r_2(\sigma_2)\quad \sigma_2\in \Sigma_2\\
 & \text{Eqs.}(\ref{sqconstraint1ad})-(\ref{sqconstraint2ad})\quad \forall i\in T\\
 &  r_2(\sigma_2)\in\{0,1\} \quad \forall  \sigma_2\in \Sigma_2\\
&  r_{1}(\sigma_{1})\geq 0 \quad  \forall \sigma_{1}\in \Sigma_{1},
 \end{align*}
where $M$ is the maximum payoff of the team,  $u(\sigma_2)=\sum_{\sigma_1\in \Sigma_1}u^{\omega_n}_{\sigma_1,\sigma_2}r_1(\sigma_1)$,  $w(\sigma_2)$ can be rewritten as a multilinear term according to its definition, i.e., $w(\sigma_2)=r_2(\sigma_2)u(\sigma_2)$, $r_1(\sigma_1)\in[0,1]$ (representing the realization-form strategy of player 1) and $r_2(\sigma_2)\in\{0,1\}$ are the probabilities when the corresponding sequence is played under the sequence-form strategy, $u^{\omega_n}_{\sigma_1,\sigma_2}$ is the known utility given the adversary realization-form strategy $\omega_n$ and sequences $\sigma_1$ and $\sigma_2$ of the team. For games with more than three players, sequence $\sigma_2$ is replaced by a joint sequence  $\times_{i\in T\setminus\{1\}}\sigma_i$. Suppose we have the following two additional bilinear terms:
  \begin{align*}
  &w(\sigma_2a)=r_2(\sigma_2a)u(\sigma_2a)\\ &w(\sigma_2b)=r_2(\sigma_2b)u(\sigma_2b),\end{align*} and the following  constraint of the sequence-form strategy: \begin{align*}r_2(\sigma_2)=r_2(\sigma_2a)+r_2(\sigma_2b).\end{align*}
We are not able to generate the associated constraint $w(\sigma_2)=w(\sigma_2a)+w(\sigma_2b)$ because the utilities  (i.e., $u(\sigma_2), u(\sigma_2a)$, and $u(\sigma_2b)$) in these terms usually differ. Consequently, the BRO of Farina  et al. \citeyearpar{farina2018ex} is not compatible with associated constraints. The reason is that when the utility is included in multilinear terms,   the remaining parts (e.g.,  $u(\sigma_2), u(\sigma_2a)$, and $u(\sigma_2b)$)) in different terms are usually not the same if one variable (e.g., $r_2(\sigma_2)$, $r_2(\sigma_2a)$, or $r_2(\sigma_2b)$) in each term is fixed. We hence can not directly generate associated constraints to reduce the feasible solution space of this BRO. In contrast, the multilinear terms (i.e., $\prod_{i\in T}r_i(\text{seq}_i(l))$) in our BRO formulation  (problem \eqref{TMECorBRO}), only represent the   team's reaching probabilities for terminal nodes, and are thus compatible with associated constraints.

Another drawback of this approach is that player 1's  sequence-form strategy computed by the BRO has to be transformed into an equivalent realization-form strategy in their FTP \cite{farina2018ex}.
Moreover, one additional advantage of our ART is that unlike the presented existing approaches, our BRO does not require that all payoffs of the team are positive.

\subsection{McCormick Relaxation}\label{RelatedWorkTMECorPartMcCormick}
The McCormick relaxation \cite{mccormick1976computability} is a technique for approximating a bilinear term $w=y_iy_j$ ($y_i,y_j\in [0,1]$) by loose inequalities $\max\{0,y_i+y_j-1\}\leq w\leq \min\{y_i,y_j\}$. The recursive version \cite{ryoo2001analysis} of this approach uses a new variable to replace each bilinear part in the multilinear term recursively. For example, the multilinear term $w'=y_iy_jy_h$ is recursively defined as two bilinear terms $w'=wy_h$ and $w=y_iy_j$, and then the McCormick inequalities are generated for each bilinear term separately. A new technique inspired by this approach \cite{zhang2020computing} adds integer variables to increase precision by generating tighter bounds.

\section{Proofs}\label{TMECorProofs}
\setcounter{Proposition}{0}
\setcounter{Lemma}{0}
\setcounter{Theorem}{0}
\setcounter{Corollary}{0}
\begin{Lemma}
For every normal-form strategy $x_T\in\Delta(\Pi_T)$ there exists a realization-equivalent hybrid-form strategy $\overline{x}_T\in\Delta(\mathcal{F}_T)$ such that $|S_{x_T}| = |S_{\overline{x}_T}|$.
\end{Lemma}
\begin{proof}
As we mentioned in Section \ref{TMECorPre}, any pure normal-form strategy $\pi_{1}\in\Pi_{1}$ can be uniquely  represented by a   pure sequence-form strategy $r_{1}\in \overline{\mathcal{R}}_1$, i.e., any pure normal-form strategy $\pi_{1}\in\Pi_{1}$ is realization-equivalent to a pure sequence-form strategy $r_{1}\in \overline{\mathcal{R}}_1$. Obviously, if $\pi_{1} \sim r_{1},$ \begin{align*}(\pi_{1},\pi_{T\setminus\{1\}})\sim (r_{1},\pi_{T\setminus\{1\}}),    \end{align*} where $\pi_1\in \Pi_1$ and $r_{1}\in \overline{\mathcal{R}}_1$. For any mixed normal-form strategy $ x_T \in \Delta(\Pi_T)$, we define a mixed hybrid-form strategy $\overline{x}_T\in \Delta(\mathcal{F}_T)$ such that: \begin{align*}&\overline{x}_T(r_{1},\pi_{T\setminus\{1\}})=x_T(\pi_{1},\pi_{T\setminus\{1\}}) \\ &\quad\quad \quad\quad\quad\quad
\forall (\pi_{1},\pi_{T\setminus\{1\}})\in \Pi_T, \pi_{1} \sim r_{1}, r_1\in \overline{\mathcal{R}}_1,\end{align*} and $\overline{x}_T(f_T) = 0$ for other strategies $f_T \in \mathcal{F}_T\setminus \{(r_{1},\pi_{T\setminus\{1\}})\mid r_1\in \overline{\mathcal{R}}_1\}$.
 Given any adversary strategy $r_n$ (or $\sigma_n$) and any node  $h\in H\cup L$, let $P_{\pi_T}(h)$ be $\pi_T$'s reaching probability for $h$. Then $x_T$'s reaching probability for $h$ is: \begin{align*} &\sum_{(\pi_{1},\pi_{T\setminus\{1\}})\in \Pi_T}\!\!\!\!x_T(\pi_{1},\pi_{T\setminus\{1\}})P_{(\pi_{1},\pi_{T\setminus\{1\}})}(h)\\
= &\sum_{(r_{1},\pi_{T\setminus\{1\}})\in \mathcal{F}_T,r_1\sim\pi_1  }\!\!\!\!\overline{x}_T(r_{1},\pi_{T\setminus\{1\}})
  P_{(\pi_{1},\pi_{T\setminus\{1\}})}(h),\end{align*} which is $\overline{x}_T$'s reaching probability for $h$.  Hence, $x_T \sim \overline{x}_T$ and $|S_{x_T}| = |S_{\overline{x}_T}|$.
\end{proof}

\begin{Lemma}
For every hybrid-form strategy $\overline{x}_T\in\Delta(\mathcal{F}_T)$ there exists a realization-equivalent normal-form strategy $x_T\in\Delta(\Pi_T)$.
\end{Lemma}
\begin{proof}
As we mentioned in Section \ref{TMECorPre},   any sequence-form strategy $r_1\in  \mathcal{R}_1$ is realization-equivalent to a mixed  normal-form strategy $x_{1}\in \Delta(\Pi_{1})$. Obviously,  if $ x_{1} \sim r_{1}$, \begin{align*}(x_1,\pi_{T\setminus\{1\}})\sim (r_1,\pi_{T\setminus\{1\}}),  \end{align*} where $x_{1}\in \Delta(\Pi_{1})$ and $r_1\in  \mathcal{R}_1$.
For any   mixed hybrid-form strategy $\overline{x}_T\in \Delta(\mathcal{F}_T)$, we define a mixed normal-form strategy $ x_T \in \Delta(\Pi_T)$ such that: \begin{align*}&x_T(\pi_{1},\pi_{T\setminus\{1\}})= \!\!\!\sum_{(r_{1},\pi_{T\setminus\{1\}})\in \mathcal{F}_T,x_1\sim r_1}\!\!\!\overline{x}_T(r_{1},\pi_{T\setminus\{1\}})x_{1}(\pi_{1})
 \\ &\quad\quad\quad\quad\quad\quad\quad\quad\quad\quad\quad\quad\quad\quad\quad\forall (\pi_{1},\pi_{T\setminus\{1\}}) \in \Pi_T.\end{align*}
Given any adversary strategy $r_n$ (or $\sigma_n$) and any node  $h\in H\cup L$, let $P_{f_T}(h)$ be $f_T$'s reaching probability for $h$. We have: if $x_{1} \sim r_{1},$   \begin{align*}P_{(r_{1},\pi_{T\setminus\{1\}})}(h) =\sum_{\pi_1\in \Pi_1}x_1(\pi_1)P_{(\pi_{1},\pi_{T\setminus\{1\}})}(h)  .\end{align*}
 Then $\overline{x}_T$'s reaching probability for $h$ is: \begin{align*}&\sum_{(r_{1},\pi_{T\setminus\{1\}})\in \mathcal{F}_T}\overline{x}_T(r_{1},\pi_{T\setminus\{1\}})P_{(r_{1},\pi_{T\setminus\{1\}})}(h)
 \\=& \sum_{(r_{1},\pi_{T\setminus\{1\}})\in \mathcal{F}_T,r_1\sim x_1}\!\!\!\!\!\!\!\!\!\!\!\!\!\!\!\!\!\! \overline{x}_T(r_{1},\pi_{T\setminus\{1\}})
 \sum_{\pi_1\in \Pi_1} \!\!  x_1(\pi_1)P_{(\pi_{1},\pi_{T\setminus\{1\}})}(h)
\\=&\!\!\!\!\!\!\!\sum_{(\pi_{1},\pi_{T\setminus\{1\}})\in \Pi_T } \sum_{(r_{1},\pi_{T\setminus\{1\}})\in \mathcal{F}_T,r_1\sim x_1} \!\!\!\!\!\!\!\!\!\!\!\!\!\!\!\!\!\!\!\!  \overline{x}_T (r_{1},\pi_{T\setminus\{1\}})x_1(\pi_1)
  P_{(\pi_{1},\pi_{T\setminus\{1\}})}(h)
  \\
  =&\sum_{(\pi_{1},\pi_{T\setminus\{1\}})\in \Pi_T } x_T (\pi_{1},\pi_{T\setminus\{1\}})
  P_{(\pi_{1},\pi_{T\setminus\{1\}})}(h),\end{align*} which is $ x_T$'s reaching probability for $h$.
 Hence, $x_T \sim \overline{x}_T$. 
\end{proof}

\begin{Theorem}
For every normal-form strategy $x_T\in \Delta(\Pi_T)$ there exists a hybrid-form strategy $\overline{x}_T\in\Delta(\mathcal{F}_T)$ such that $U_T(x_T,\sigma_n)=U_T(\overline{x}_T,\sigma_n)$, $\forall \sigma_n\in\Sigma_n$, and vice versa.
\end{Theorem}
\begin{proof}
By Lemma \ref{realizationEQLemma}, $\forall x_T\in \Delta(\Pi_T)$, $\exists \overline{x}_T\in\Delta(\mathcal{F}_T)$ such that $x_T\sim\overline{x}_T$. For each  $\sigma_n\in\Sigma_n$, let $p_{\sigma_n}(l)$ be the probability reaching terminal node $l$ by $x_T$ or $\overline{x}_T$.
Then,
 \begin{align*}&U_T(\overline{x}_T,\sigma_n)\\=&\sum_{f_T\in \mathcal{F}_T}U_T(f_T,\sigma_n)\overline{x}_T(f_T)\\=& \sum_{l\in L}p_{\sigma_n}(l)u_T(l)c(l)\\=&\sum_{\pi_T\in \Pi_T}U_T(\pi_T,\sigma_n)x_T(\pi_T)\\=& U_T(x_T,\sigma_n).\end{align*} Similarly, by Lemma \ref{realizationEQ}, $\forall \overline{x}_T\!\!\in\!\! \Delta(\mathcal{F}_T)$, $\exists x_T\in\Delta(\Pi_T)$ such that $U_T(x_T,\sigma_n)\!\!=\!\!U_T(\overline{x}_T,\sigma_n)(\forall \sigma_n\in\Sigma_n)$.
\end{proof}

\begin{Proposition}
CMB converges in at most $2^{|\Pi_{1}|}\frac{|\Pi_T|}{|\Pi_1|}$ steps. 
\end{Proposition}
\begin{proof}
Given the adversary strategy $r_n$, suppose that $f_T=(r_1,\pi_{T\setminus\{1\}})$ is a best response for the team. Let   $x_1$ be a realization-equivalent normal-form strategy of $r_1$ with its support set $S_{x_1} =\{\pi_1\mid x_1(\pi_1)>0, \pi_1\in \Pi_1\}$. Then, $(x_1,\pi_{T\setminus\{1\}})$ and $f_T$ are realization-equivalent with: \begin{align*}&U_T(f_T,r_n)\\=&U_T((x_1,\pi_{T\setminus\{1\}}),r_n)\\=&\sum_{\pi_1\in S_{x_1}}x_1(\pi_1)U_T(\pi_1,\pi_{T\setminus\{1\}},r_n).\end{align*} Note that, if $(x_1,\pi_{T\setminus\{1\}})$   is a best response against $r_n$, then $(\pi_1,\pi_{T\setminus\{1\}})(\forall \pi_1\in S_{x_1})$ is also a best response against $r_n$. 
Then, for any $\pi_1\in S_{x_1}$  we   have $U_T(\pi_1,\pi_{T\setminus\{1\}},r_n)=U_T(f_T,r_n)$.\footnote{If there is $  \pi_1\in S_{x_1}$ such that $U_T(\pi_1,\pi_{T\setminus\{1\}},r_n)<U_T((x_1,\pi_{T\setminus\{1\}}),r_n)$, then there is $\pi'_1\in S_{x_1}$ such that $U_T(\pi'_1,\pi_{T\setminus\{1\}},r_n)>U_T((x_1,\pi_{T\setminus\{1\}}),r_n)$, and then  higher $U_T((x_1,\pi_{T\setminus\{1\}}),r_n)$ is achieved by moving the probability for $\pi_1$ to $\pi'_1$, which causes a contradiction.} 
Therefore, for any $
r'_1\in \mathcal{R}_1$ with a realization-equivalent normal-form strategy  $x'_1$ such that $S_{x'_1}=S_{x_1}$, \begin{align*}&U_T(r'_1,\pi_{T\setminus\{1\}},r_n)\\=&\sum_{\pi_1\in S_{x_1}}\!\!\!x'_1(\pi_1)U_T(\pi_1,\pi_{T\setminus\{1\}},r_n)\\=&U_T(x_1,\pi_{T\setminus\{1\}},r_n)\\ =&U_T(f_T,r_n).\end{align*} 
 For each $\Pi'_1\subseteq \Pi_1$, we can define a set of pure hybrid-form strategies: \begin{align*}&\mathcal{F}_T(\Pi'_1,\pi_{T\setminus\{1\}})\\=&\{(r_1,\pi_{T\setminus\{1\}})\in\mathcal{F}_T\mid S_{x_1} = \Pi'_1(\subseteq \Pi_1), x_1 \sim r_1  \}.\end{align*} where $\forall f_T\in \mathcal{F}_T(\Pi'_1,\pi_{T\setminus\{1\}})$, $f_T$ is a best response against $r_n$ if $\exists$ $f'_T \in \mathcal{F}_T(\Pi'_1,\pi_{T\setminus\{1\}})$ is a best response against $r_n$. In CMB, only the best response against $r_n$, which must be better than any strategy in  $\mathcal{F}'_T$ against $r_n$,  will be added to $\mathcal{F}'_T$. Therefore, only one strategy $f'_T$ in $\mathcal{F}_T(\Pi'_1,\pi_{T\setminus\{1\}})$ will be added to $\mathcal{F}'_T$, i.e., the set of pure hybrid-form strategies $\mathcal{F}_T(\Pi'_1,\pi_{T\setminus\{1\}})$ is represented by a single strategy $f'_T$ in $\mathcal{F}'_T$. In the worst case, CMB needs to add strategies involving all these sets (the number of these sets is up to  $2^{|\Pi_{1}|}\prod_{i\in T\setminus\{1\}}|\Pi_i|$), and then CMB will terminate within $2^{|\Pi_{1}|}\frac{|\Pi_T|}{|\Pi_1|}$ 
 iterations.
\end{proof}

\begin{Proposition}
There is a  TMECor $(\overline{x}_T,r_n)$  with    $|S_{\overline{x}_T}|\leq |\Sigma_n|$. 
\end{Proposition}
\begin{proof}
There is  TMECor ($x_T,r_n$)   with $|S_{x_T}|\leq|\Sigma_n|  $ \cite{celli2018computational}.  We know that there is $\overline{x}\in \Delta(\mathcal{F}_T)$ with $|S_{x_T}| = |S_{\overline{x}_T}|$    such that $x_T\sim\overline{x}_T$  by Lemma \ref{realizationEQLemma}. By Theorem \ref{coroTMECorEQU}, $U_T(x_T,\sigma_n)=U_T(\overline{x}_T,\sigma_n)(\forall \sigma_n\in\Sigma_n)$. Then,  $(\overline{x}_T,r_n)$ is a TMECor, concluding the proof.
%
\end{proof}

\begin{Theorem}
  CMB converges to a TMECor in   $(\mathcal{F}_T,\Sigma_n)$.
\end{Theorem}
\begin{proof}
By Proposition \ref{upperboundofiteration}, CMB will terminate with $\overline{v}=\underline{v}$ within a finite number of iterations. Then, we have, $U_T(\overline{x}_T,r_n)= \underline{v}=\overline{v}\geq U_T(\overline{x}'_T,r_n)(\forall \overline{x}'_T)$, and $-U_T(\overline{x}_T,r_n)= -\underline{v} \geq -U_T(\overline{x}_T,r'_n)(\forall r'_n)$. Therefore, $(\overline{x}_T,r_n)$ is a TMECor.
\end{proof}

\begin{Proposition}
If CMB terminates with   $\overline{v}-\underline{v}\leq \epsilon$, then its output $(\overline{x}_T,r_n)$ is an $\epsilon$-TMECor in $(\mathcal{F}_T,\Sigma_n)$.
\end{Proposition}
\begin{proof}
$U_T(\overline{x}'_T,r_n)-U_T(\overline{x}_T,r_n)\leq \overline{v} -\underline{v}\leq \epsilon$ $(\forall \overline{x}'_T)$, and $-U_T(\overline{x}_T,r'_n)-(-U_T(\overline{x}_T,r_n))\leq  -\underline{v} -(-\underline{v})= 0 (\forall r'_n)$. Therefore, $(\overline{x}_T,r_n)$ is an $\epsilon$-TMECor.
\end{proof}

\begin{Proposition}
$\prod_{ i\in T }r_i(\sigma_T(i))$ with $r_i\in \overline{\mathcal{R}}_i $ for all $i\in T\setminus\{1\}$ and $r_1\in \mathcal{R}_1$ is exactly represented by $w(\sigma_T)$ in Eqs.(\ref{werconstraint1}) and  (\ref{werconstraint2}).
\end{Proposition}
\begin{proof}
 $\prod_{ i\in T }r_i(\sigma_T(i))= r_1( \sigma_T( 1) )$ if $r_i(\sigma_T(i))=1(\forall i\in T\setminus\{1\})$, otherwise $\prod_{ i\in T }r_i(\sigma_T(i))=0$. First, by Eq.(\ref{werconstraint2}),  $0\leq  r_1( \sigma_T( 1) )- w(\sigma_T)\leq n-2-(n-2)=0$ if $r_i(\sigma_T(i))=1(\forall i\in T\setminus\{1\})$, i.e., $w(\sigma_T)= r_1( \sigma_T( 1) )$. Second, by Eq.(\ref{werconstraint1}), $w(\sigma_T)=0$ if any $i\in T\setminus\{1\}$ with $r_i(\sigma_T(i))=0$. Then, $\prod_{ i\in T }r_i(\sigma_T(i))$   is exactly represented by $w(\sigma_T)$ in Eqs.(\ref{werconstraint1}) and  (\ref{werconstraint2}). 
\end{proof}
\begin{Theorem}
The solution of problem (\ref{TMECorBRO}) solves problem (\ref{TMECorBROART}). 
\end{Theorem}
 \begin{proof}
Suppose that a sequence-form strategy $r_i$   of player $i$  is   feasible in Problem   (\ref{TMECorBRO})    but  is infeasible in  Problem (\ref{TMECorBROART}). 
By  Proposition \ref{MultilinearRepresentation}, $w(\times_{i\in T}\text{seq}_i(l))$ is exactly represented by Eqs.(\ref{werconstraint1}) and  (\ref{werconstraint2}).
Therefore, $r_i$  does not satisfy the constraints   in Eq.(\ref{associatedConstraintEQ}).
 Without loss of generality, we assume that  $\sum_{a\in \chi(I_{i,j})} r_i(\sigma_{T}(i) a)=r_i(\sigma_{T}(i))$,  but $\sum_{a\in \chi(I_{i,j})} w(\sigma_{T}(i) a, \sigma_{T\setminus\{i\}})\neq w(\sigma_{T})$.  We have:
\begin{align*}  &\textstyle \sum_{a\in \chi(I_{i,j})} w(\sigma_{T}(i) a, \sigma_{T\setminus\{i\}})\neq w(\sigma_{T})  \\
\Rightarrow &    \!\!\!\!\!\!\!   \sum_{a\in \chi(I_{i,j})}    \!\!\!\!\!\!\!   r_i(\sigma_{T}(i) a) \!\!\!\!\!   \prod_{ j\in T\setminus\{i\} } \!\!\!\!    r_j(\sigma_T(j)) \neq   r_i(\sigma_{T}(i)  ) \!\!\!\!\!   \prod_{ j\in T\setminus\{i\} }  \!\!\!\!   r_j(\sigma_T(j))\\
\Rightarrow &\textstyle\sum_{a\in \chi(I_{i,j})} r_i(\sigma_{T}(i) a) \neq  r_i(\sigma_{T}(i)  ),
\end{align*}
which causes a contradiction. Therefore, the  feasible solution of  Problem (\ref{TMECorBRO}) is feasible in Problem (\ref{TMECorBROART}). 
\end{proof}
\begin{Corollary}
For any strategy $r_n$ of the adversary, the optimal solution of Problem (\ref{TMECorBROART}) is a best response against $r_n$. 
\end{Corollary}
\begin{proof}
By Proposition \ref{MultilinearRepresentation} and Theorem \ref{theoremFeasibleART}, the  feasible solution of  Problem (\ref{TMECorBRO}) is feasible in Problem (\ref{TMECorBROART}) (and both problems have the same sequence-form strategy space), and the objective functions in both problems are the same. Therefore, the   optimal solution of Problem (\ref{TMECorBROART}) is   optimal in Problem (\ref{TMECorBRO}), concluding the proof.
 \end{proof}
 \section{Associated Constraints of Kuhn Poker}
   \label{Example4pKuhnInf}
In a four-player Kuhn poker game, the information set for player $i$ is defined by the card that player $i$ holds and the observed actions taken by players in turn.
For example, \text{\it J:/cccr:} is an information set of player 1 reachable by sequence $\text{\it J:/:c}$, in which player 1 has a card  \text{\it J}  and observed the  actions   \text{\it `c,c,c,r'} of all players in turn.  There are two possible actions player 1 can take: \text{\it call} (\text{\it c}) 
and \text{\it fold} (\text{\it f}). Therefore, by Eq.(\ref{sqconstraintsecondad}), we have
  $r_1(\text{\it J:/:c})=r_1(\text{\it J:/cccr:c})+r_1(\text{\it J:/cccr:f})$. 
This information set may contain many possible nodes because of the uncertainty in what cards of other players hold. Each node in the information set corresponds to a sequence for each player from the root to this node. For example $(\text{\it J:/:c, Q:/c:c, T:/cc:c, K:/ccc:r})$ and $(\text{\it J:/:c, Q:/c:c, K:/cc:c, T:/ccc:r})$ are two different nodes in this information set with different cards of players 3 and 4. Apparently, there are many different joint sequences for the team to reach an information set,   corresponding to  different nodes in this information set, e.g.,  $(\text{\it J:/:c, Q:/c:c, T:/cc:c})$ or $(\text{\it J:/:c, Q:/c:c, K:/cc:c})$.  
Because player 1 acts in this information set and chooses from two possible actions, for given a joint sequence $(\text{\it J:/:c, Q:/c:c, T:/cc:c})$ that leads to this information set, there exist the extended joint sequences $(\text{\it J:/cccr:c, Q:/c:c, T:/cc:c})$ and $(\text{\it J:/cccr:f, Q:/c:c, T:/cc:c})$ that lead to some immediately succeeding information sets. Therefore, we can generate the following associated constraint (note that $w(\sigma_T)= \prod_{ i\in T }r_i(\sigma_T(i))$):
 \begin{align*}&w(\text{\it J:/:c, Q:/c:c, T:/cc:c})\\=&w(\text{\it J:/cccr:c, Q:/c:c, T:/cc:c})\\&+w(\text{\it J:/cccr:f, Q:/c:c, T:/cc:c}),\end{align*}
 because of $r_1(\text{\it J:/c:c})=r_1(\text{\it J:/cccr:c})+r_1(\text{\it J:/cccr:f})$. Moreover, given $w(\text{\it J:/:c, Q:/c:c, T:/cc:r})$, we have the following associated constraint:
  \begin{align*}&w(\text{\it J:/:c, Q:/c:c}, \varnothing)\\=&w(\text{\it J:/:c, Q:/c:c, T:/cc:r})\\&+w(\text{\it J:/:c, Q:/c:c, T:/cc:c}) .\end{align*}
   In addition, given $w(\text{\it J:/:c, Q:/c:r}, \varnothing)$, the following associated constraint holds:   \begin{align*}&w(\text{\it J:/:c, Q:/c:c}, \varnothing)+w(\text{\it J:/:c, Q:/c:r}, \varnothing)\\=&w(\text{\it J:/:c}, \varnothing, \varnothing)\\=&r_1(\text{\it J:/:c}).\end{align*} 
  Note that $w(\varnothing,\varnothing,\varnothing)=r_1(\varnothing)=1.$
  
Assume that the team's joint sequence   $(\text{\it J:/cccr:c, Q:/c:c, T:/cc:c})$ reaches a node $(\text{\it J:/cccr:c, Q:/c:c, T:/cc:c, K:/ccc:r})$ in player 2's information set $\text{\it Q:/cccrc:}$. Because player 2 can also choose from     two actions $c$ and $f$, we can generate an associated constraint:
\begin{align*}&w(\text{\it J:/cccr:c, {Q:/c:c}, T:/cc:c})\\ =&w(\text{\it J:/cccr:c,{Q:/cccrc:c},T:/cc:c}) \\&+w(\text{\it J:/cccr:c, {Q:/cccrc:f},   T:/cc:c}).\end{align*}
Subsequently, when the team's joint sequence $(\text{\it J:/cccr:c,{Q:/cccrc:c},T:/cc:c})$ reaches a node $(\text{\it J:/cccr:c, Q:/cccrc:c,}$ $\text{T:/cc:c, K:/ccc:r})$ in player 3's information set $\text{\it  T:/cccrcc:}$ (also with   two actions $c$ and $f$), we have an associated constraint:
 \begin{align*}
     &w(\text{\it J:/cccr:c, Q:/cccrc:c, {T:/cc:c}}) \\
     =&w(\text{\it J:/cccr:c, Q:/cccrc:c, {T:/cccrcc:c}})\\&+w(\text{\it J:/cccr:c, Q:/cccrc:c,  {T:/cccrcc:f}}).
 \end{align*}  
 Here, nodes $(\text{\it J:/cccr:c, Q:/cccrc:c, {T:/cccrcc:c}, K:/ccc:r})$ and $w(\text{\it J:/cccr:c, Q:/cccrc:c, {T:/cccrcc:f},}$ $\text{K:/ccc:r})$ are terminal nodes, and the team's joint sequence $(\text{\it J:/cccr:c, Q:/cccrc:c, {T:/cccrcc:c}})$ is considered in Problem (\ref{TMECorBRO}) as a joint sequence reaching a terminal node. Once we enumerate all team members' information sets according to Eq.(\ref{associatedConstraintEQ}), we add all associated constraints corresponding to these joint sequences reaching terminal nodes to Problem (\ref{TMECorBROART}).

\section{Definitions of Evaluation Domains}\label{pokergamerules}
The rules of these games are according to the standard setting \cite{abou2010using,farina2018ex}.
\subsection{Kuhn Poker Games}
$n$K$r$  includes $n$ players and   $r$ possible cards. Each player pays one chip to the pot before the cards are dealt, and then is dealt one private card. The game players take actions in turns. There is   one betting round with one bet, i.e., adding one chip to the pot. The first player may bet or check. After that, a player can check (call) or bet (raise) if there is no bet; otherwise, a player can call or fall if there is a bet. At the showdown, the player who has not folded and has the highest card wins all chips.

\subsection{Leduc Poker Games}
$n$L$r$  includes $n$ players and   $r$ ranks (i.e., $3r$   cards). The initial setting is the same as the one in Kuhn poker games. However, there are two betting rounds, the first one is called preflop (with the preflop bet: two chips), and the second one is called flop (with the flop bet: four chips). Each betting round is the same as the one in Kuhn poker games. After the preflop betting, one community card is dealt (all players can see its information), then the flop betting begins. After the flop betting, players who did not fold showdown (reveal the private cards). The player who pairs the community card with her private card wins the pot; otherwise, the player who has the highest card wins. If some players have the same card and win, they  split the pot.

Our poker games are implemented according to \url{https://github.com/achao2013/cppcfr}. In Leduc games of this implementation, players do not distinguish the different suits of cards for sequences, but players   distinguish the suits of cards for nodes. If players distinguish the suits of cards for nodes, the number of terminal nodes may be fewer, e.g., this number is 6477  in the 3L3 game.

\end{document}